\documentclass{IEEEtran}

\usepackage{amsmath,amssymb,amsfonts,bm}
\usepackage{algorithmic}
\usepackage{graphicx}
\usepackage{float}
\usepackage{mdwmath}

\usepackage{amsthm}
\usepackage{mathrsfs}

\usepackage{textcomp}
\usepackage{subfigure}
\usepackage{subfig}

\usepackage{soul,color}
\usepackage{cancel}
\usepackage{cite}

\theoremstyle{plain}
\newtheorem{thm}{Theorem}

\newtheorem{prop}[thm]{Proposition}

\newtheorem{defn}[thm]{Definition}

\newcommand{\norm}[1]{\left\lVert#1\right\rVert}

\begin{document}

\title{LDA-MIG Detectors for Maritime Targets in Nonhomogeneous Sea Clutter}

\author{Xiaoqiang~Hua,~
        Linyu~Peng,~\IEEEmembership{Member,~IEEE,}
        Weijian~Liu,~\IEEEmembership{Senior Member,~IEEE,}
        Yongqiang~Cheng,~
        Hongqiang~Wang,~
        Huafei~Sun,~
        and Zhenghua~Wang
\thanks{This work was partially supported by grants from NSFC  (Nos. 61901479, 62071482, 62035014 and 61921001), JSPS KAKENHI (No. JP20K14365), JST CREST (No. JPMJCR1914), and Keio Gijuku Fukuzawa Memorial Fund.  {\it (Corresponding author: Linyu Peng.)}

X. Hua is with the College of Electronic Science, and also with the College of Computer, National University of Defense Technology, Changsha 410073, China (e-mail: hxq712@yeah.net).

L. Peng is with the Department of Mechanical Engineering, Keio University, Yokohama 223-8522, Japan (e-mail: l.peng@mech.keio.ac.jp).

W. Liu is with the Wuhan Electronic Information Institute, Wuhan 430019, China (e-mail: liuvjian@163.com).

Y. Cheng and H. Wang are with the College of Electronic Science, National University of Defense Technology, Changsha 410073, China (e-mail: cyq101600@126.com; oliverwhq@163.net).

H. Sun is with the School of Mathematics and Statistics, Beijing Institute of Technology, Beijing 100081, China (email: huafeisun@bit.edu.cn).

Z. Wang is with the College of Computer, National University of Defense Technology, Changsha 410073, China (email: wzh123202209@163.com).
}}


\maketitle

\begin{abstract}
This paper deals with the problem of detecting maritime targets embedded in nonhomogeneous sea clutter, where limited number of secondary data is available due to the heterogeneity of sea clutter.
A class of linear discriminant analysis (LDA)-based matrix information geometry (MIG) detectors is proposed in the supervised scenario. As customary, Hermitian positive-definite (HPD) matrices are used to model the observational sample data, and  the clutter covariance matrix of received dataset is estimated as geometric mean of the secondary HPD matrices. Given a set of training HPD matrices with  class labels, that are elements of a higher-dimensional HPD matrix manifold, the LDA manifold projection learns a mapping from the higher-dimensional HPD matrix manifold to a lower-dimensional one subject to maximum discrimination. In the current study, the LDA manifold projection, with the cost function maximizing between-class distance while minimizing within-class distance, is formulated as an optimization problem in the Stiefel manifold. Four robust LDA-MIG detectors corresponding to different geometric measures are proposed. Numerical results based on both simulated radar clutter with interferences and real IPIX radar data show the advantage of the proposed LDA-MIG detectors against their counterparts without using LDA as well as the state-of-art maritime target detection methods in nonhomogeneous sea clutter.
\end{abstract}

\begin{IEEEkeywords}
Target detection, matrix information geometry (MIG) detectors, linear discriminant analysis (LDA), HPD manifold, nonhomogeneous sea clutter.
\end{IEEEkeywords}

\section{INTRODUCTION}
D{\scshape etecting} maritime targets within nonhomogeneous sea clutter is an imperative problem in radar remote sensing community \cite{7987714,9016366,6479294,8617694,6517933,6915876,9336010,9638633,Chen2000611}. Due to the nonhomogeneity of sea clutter arising in clutter discretes, clutter edge, etc., estimation of the clutter covariance matrix is often based on limited number of homogeneous secondary data. Insufficiency of secondary data often leads to a degradation in the detection performance. To ameliorate this problem, many efforts have been concentrated on designing knowledge-based detectors that involve {\it a priori} information on the nonhomogeneous sea clutter \cite{9460779,8721549,4014433,7426844,8995791}. For instance, {\it a priori} information about the topography of the observed scene was exploited to remove dynamic outliers from the training data  in \cite{4014433}. In \cite{4359541}, a special covariance structure that was modeled as the sum of a positive semi-definite matrix plus a term proportional to the identity was incorporated into the estimation of the disturbance covariance matrix to achieve detection performance improvement. Another example was given in \cite{7426844}, in which the Wald test, the Rao test, and the generalized likelihood ratio test were designed based upon the assumption of a symmetrically structured power spectral density to estimate the clutter covariance matrix; the experimental results, performed on simulation clutter and real data, confirmed the advantage of the proposed detectors over their conventional counterparts (see also \cite{4267625,5210021,5484507,9145700,9490335,9635684,9856933,9749061}). In \cite{9336010,9638633}, the authors proposed several excellent detectors for maritime targets embedded in nonhomogeneous sea clutter by resorting to the deep CNN, which significantly improved the detection performance if sufficient sample data is available.
However, the aforementioned methods heavily depend on the clutter characteristics that is not easy to be captured sufficiently. Lack of environment information often results in poor detection performance.

Recently, a new signal detection technique, referred to as the affine invariance Riemannian metric (AIRM)-based matrix information geometry (MIG) detector, was proposed in \cite{4721049}. This detector does not require any {\it a priori} information about clutter characteristics but simply takes the intrinsic Riemannian geometry of the resulting manifold into account. In the AIRM-MIG detector, Hermitian positive-definite (HPD) matrices are used to capture the correlation or power of sample data, and the clutter covariance matrix is estimated as the Riemannian mean of secondary HPD matrices. A decision about the absence or presence of a target signal can be made by comparing the AIRM distance between the HPD matrix in the cell under test (CUT) and the clutter covariance matrix with a given threshold. Results on the analysis of clutter characteristics \cite{9078971}, the wake monitoring \cite{Liu2013,CO201054}, the scatter matrix estimation \cite{7842633}, and the target detection in X-band radars \cite{7226283,6514112,4720937} have shown that the AIRM mean is robust and the AIRM-MIG detector outperforms the classical detector. Beside its applications to target detection, MIG has also been applied in the training sample selection, the clutter covariance matrix estimation \cite{7060458} and the space-time adaptive processing \cite{6450651}. As a matter of fact, outstanding performance of the AIRM-MIG detector is a consequence of the excellent discriminative power of the geodesic distance on the HPD matrix manifold.

In addition to the AIRM, many geometric measures can be defined on HPD matrix manifolds, that reflect different local geometric structures and hence yield different discriminative power. A straightforward idea to improve the detection performance is to exploit discriminative geometric measures to design the MIG detectors. This guideline has been implemented in \cite{8000811,HUA2017106,HUA2018232,HUA2076,9764734}, where the authors designed MIG detectors by resorting to different geometric measures and analyzed their differences in the detection performance. As expected, MIG detectors induced by different geometric measures have exhibited different detection performances. Natural questions then arise: {\it {Which measure can lead to better performance and how to choose a suitable measure?}} One answer is by studying the corresponding anisotropy factors, as, evidentially, target detection via MIG detectors in the K-distributed clutter shows that a geometric measure with larger anisotropy value can lead to better detection performance \cite{8000811}. However, the discriminative power of an HPD matrix manifold equipped with a geometric measure depends not only  on the matrix structure but also on the clutter characteristics. It is practically difficult to choose a geometric measure that can always possess good discrimination ability, since the clutter characteristics usually vary in the domain of space and time.

This paper addresses these drawbacks by proposing a class of linear discriminant analysis (LDA)-MIG detectors in nonhomogeneous sea clutter. The LDA is an effective method for reducing redundant information in the sample data as well as for increasing the class separability. It has been widely applied in computer vision \cite{9409780}, image classification \cite{8720003}, signal processing \cite{CHEBBI20121872}, etc.
Different from the previous MIG detectors, the proposed LDA-MIG detectors can achieve enhanced discriminative power of geometric measures on HPD manifolds in different scenarios by learning a manifold projection from a training dataset, which results in significant performance improvements. Main contributions of this paper are briefly summarized below.

\begin{enumerate}
  \item Inspired by the LDA, for an HPD matrix manifold endowed with a given geometric measure, we propose a LDA manifold projection that aims at improving the discriminative power by learning a mapping from the HPD matrix manifold  to a more discriminative lower-dimensional manifold which maximizes the between-class distance while minimizes the within-class distance. Specifically, HPD matrices are employed to model the sample data and the geometric mean of secondary HPD matrices is used as the estimate of the clutter covariance matrix. Given a set of training HPD matrices with class labels, we model a LDA manifold projection as the cost function that aims to achieve maximum discrimination by resorting to the rule subject to an orthonormal constraint; thus, searching for such a projection mapping HPD matrices from a higher-dimensional manifold to a lower-dimensional one can be formulated as an optimization problem on a Stiefel manifold, which can be solved by the Riemannian gradient descent algorithm. Given a set of training HPD matrices, the projection matrix can be unique obtained.
  \item Four LDA-MIG detectors are proposed on the projected  HPD manifold by employing the AIRM  \cite{pennec2006,Bha2009,sun2016elementary,moakher2011the}, Log-Euclidean metric (LEM) \cite{AFPA2007}, Jensen--Bregman LogDet divergence (JBLD) and symmetrized Kullback--Leibler divergence (SKLD)\cite{Moakher816} as the geometric measures. The manifold projection is incorporated into the detection design. Moreover, influence functions are defined for analyzing the robustness of geometric mean to outliers, and numerical results confirm their robustness.
  \item Simulation and real radar data demonstrate that the proposed detectors have better performance than their counterparts without the LDA and the state-of-art maritime target detectors in nonhomogeneous sea clutter.
\end{enumerate}

The rest of the paper is organized as follows. The methodology of LDA-MIG detectors is formulated in Section \ref{sec:pf}. Section \ref{sec:MIG} introduces briefly the mathematical knowledge of MIG. In Section \ref{sec:SMP}, we define the projection between a higher-dimensional HPD matrix manifold to a more discriminative lower-dimensional one subject to an orthonormal constraint as an optimization problem on a Stiefel manifold, and then solve it by using the Riemannian gradient descent algorithm. Influence functions with respect to various geometric measures are derived in closed-form for analyzing their robustness to outliers. Performance analysis based on the simulation and real radar data is provided in Section \ref{sec:ns}, followed with conclusions in Section \ref{sec:con}.

\textit{Notations:} Vectors and matrices are denoted by boldface lowercase and boldface uppercase letters, respectively. $\operatorname{det}(\cdot)$, $\operatorname{tr}(\cdot)$ and $\norm{ \cdot }$ stand for the determinate, trace and Frobenius norm of a matrix, respectively. Superscripts $(\cdot)^{\operatorname{T}}$ and $(\cdot)^{\operatorname{H}}$ denote the transpose and conjugate transpose of matrices or vectors, respectively. The $N\times N$ identity matrix is denoted by $\bm{I}_N$ or simply $\bm{I}$ if no confusion would be caused. $N$-dimensional complex vectors and $N \times N$ complex matrices are denoted by $\mathbb{C}^N$ and $\mathbb{C}^{N\times N}$. The (compact and complex) Stiefel manifold is denoted by $\operatorname{St}(M,\mathbb{C}^N)$. That a matrix $\bm{A}$ is positive-definite is represented as $\bm{A}\succ \bm{0}$. All $N\times N$ complex HPD matrices form an HPD matrix manifold $\mathscr{P}(N,\mathbb{C})$. The imaginary unit is written as $\operatorname{i}$, and $\operatorname{E}[\cdot]$ denotes the statistical expectation.

\section{PROBLEM FORMULATION}

\label{sec:pf}
As customary, the detection problem of radar targets within sea clutter with sample data $\bm{y}=[y_0,y_1,\ldots,y_{N-1}]^{\operatorname{T}} \in \mathbb{C}^N$ collected from $N$ temporal stationary channels can be formulated as the following binary hypothesis testing problem \cite{huaetal2020}:
\begin{equation}
\left\{
\begin{aligned}
&\mathcal{H}_0: \left\{
\begin{aligned}
&\bm{y} = \bm{c} \\
&\bm{y}_k = \bm{c}_k, \quad k\in [K],
\end{aligned} \right.\\
&\mathcal{H}_1: \left\{
\begin{aligned}
&\bm{y} = \alpha \bm{s} + \bm{c}, \\
&\bm{y}_k = \bm{c}_k, \quad k=[K],
\end{aligned} \right. \\
\end{aligned} \right.
\end{equation}
where $[K]$ denotes the index set $\{1,2,\ldots,K\}$, and $\mathcal{H}_0$ and $\mathcal{H}_1$ denote the hypotheses of absence and presence of a target signal, respectively. The vectors $\bm{y}$ and $\bm{y}_k, k=[K]$ denote the observed data. The vectors $\bm{c}$ and $\bm{c}_k, k=[K]$ stand for the clutter data. The unknown and complex scalar-valued $\alpha$ accounts for the channel propagation effects and target reflectivity. The known signal steering vector $\bm{s}$ is given by
\begin{equation}
\bm{s} = \frac{1}{\sqrt{N}}[1, \exp(-\operatorname{i}2\pi f_d), \ldots, \exp(-\operatorname{i}2\pi f_d(N-1))]^{\operatorname{T}},
\end{equation}
with $f_d$  the normalized Doppler frequency.

In general, for a set of $K$ secondary sample data $\{ \bm{y}_k \}_{k=1}^K$, the clutter covariance matrix estimate is given by the sample covariance matrix (SCM), namely
\begin{equation}
\begin{aligned}
\bm{R}_{SCM} = \frac{1}{K}\sum_{k=1}^K \bm{y}_k\bm{y}_k^{\operatorname{H}}, \quad \bm{y}_k \in \mathbb{C}^N.
\label{eq:SCM}
\end{aligned}
\end{equation}
In other words, \eqref{eq:SCM} gives the arithmetic mean of $K$ Hermitian matrices $\{\bm{y}_k\bm{y}_k^{\operatorname{H}}\}_{k\in[K]}$ with rank one, where $\bm{y}_k\bm{y}_k^{\operatorname{H}}$ provides  the autocorrelation characteristic of the sample data $\bm{y}_k$. It was noticed that the SCM estimator is sensitive to the outlier while the geometric estimators in  HPD matrix manifolds possess good robustness \cite{huaetal2020}. Thus, we employ HPD matrices to capture the correlation or power of the sample data to discriminate the target signal and clutter.
Many matrix structures can be used for modeling the sample data, including the Toeplitz structure \cite{7060458}, the diagonal loading \cite{5417174}, the shrinkage estimators \cite{6894189}, the persymmetric covariance estimators \cite{7593261}, etc. In the current paper, we focus on the diagonal loading structure, namely
\begin{equation}
\begin{aligned}
\bm{R} = \bm{r}\bm{r}^{\operatorname{H}} + \operatorname{tr}(\bm{r}\bm{r}^{\operatorname{H}})\bm{I},
\label{eq:AR1}
\end{aligned}
\end{equation}
where
\begin{equation}
\bm{r} = [r_0, r_1, \ldots, r_{N-1}]^T,
\end{equation}
whose components are
\begin{equation}
r_l = \operatorname{E}[y_i \bar{y}_{i+l}],\quad  0\leq l \leq N-1, i\in [N - l - 1].
\end{equation}
Here, $\bar{y}$ is the conjugate of $y$ and $r_l$ is the $l$-th correlation coefficient of $\bm{y}$. The ergodicity of stationary Gaussian process allows us to approximate the correlation coefficients $r_l$ by  mean of the sample data, i.e.,
\begin{equation}
\widetilde{r}_l = \frac{1}{N} \sum_{i=0}^{N-1- l } {y_i\bar{y}_{i+l}},\quad  0 \leq l \leq N-1.
\end{equation}

The sample data, with or without target signal,  are transformed into new observations of HPD matrices by using \eqref{eq:AR1}. The HPD matrices are computed by the clutter plus target signal or the clutter, and each HPD matrix corresponds to a point in $\mathscr{P}(N,\mathbb{C})$. Given a set of secondary data $\{ \bm{R}_k \}_{k\in[K]}$ that are calculated via \eqref{eq:AR1}, the clutter covariance matrix can be estimated by the geometric mean
\begin{equation}
\bm{R}_\mathcal{G} = \mathcal{G}\left(\{ \bm{R}_k \}_{k\in[K]}\right).
\end{equation}
The essence of signal detection is to discriminate the target signal and clutter, both of which are interpreted  as HPD matrices in $\mathscr{P}(N,\mathbb{C})$.  Presumably, the HPD matrix corresponding to target signal is located relatively far away from those corresponding to clutter only. To determine the absence or presence of a target signal is then equivalent to measure the dissimilarity between the geometric mean of secondary HPD matrices and HPD matrix of the CUT.


Motivated by the LDA, in our detection framework, a LDA manifold projection that maps HPD matrices constructed from the target signal and clutter using \eqref{eq:AR1} in a higher-dimensional HPD matrix manifold to a lower-dimensional one is defined subject to an orthonormal constraint to improve the discriminative power, as shown in Fig. \ref{fig:LDA_Projection}. This mapping is defined by
\begin{equation}
\begin{aligned}
f_{\bm{W}}:\mathscr{P}(N,\mathbb{C})&\rightarrow \mathscr{P}(M,\mathbb{C})\\
\bm{R}&\mapsto \bm{W}^{\operatorname{H}}\bm{R}\bm{W},
\end{aligned}
\end{equation}
where $N>M$ and $\bm{W}$ is in the
Stiefel manifold \cite{Sti1935}
\begin{equation}\label{Stmanifold}
\operatorname{St}(M,\mathbb{C}^N)=\left\{\bm{A}\in \mathbb{C}^{N\times M} \mid \bm{A}^{\operatorname{H}}\bm{A}=\bm{I}_M\right\}.
\end{equation}

\begin{figure}[H]
  \centering
  \includegraphics[width=8.5cm]{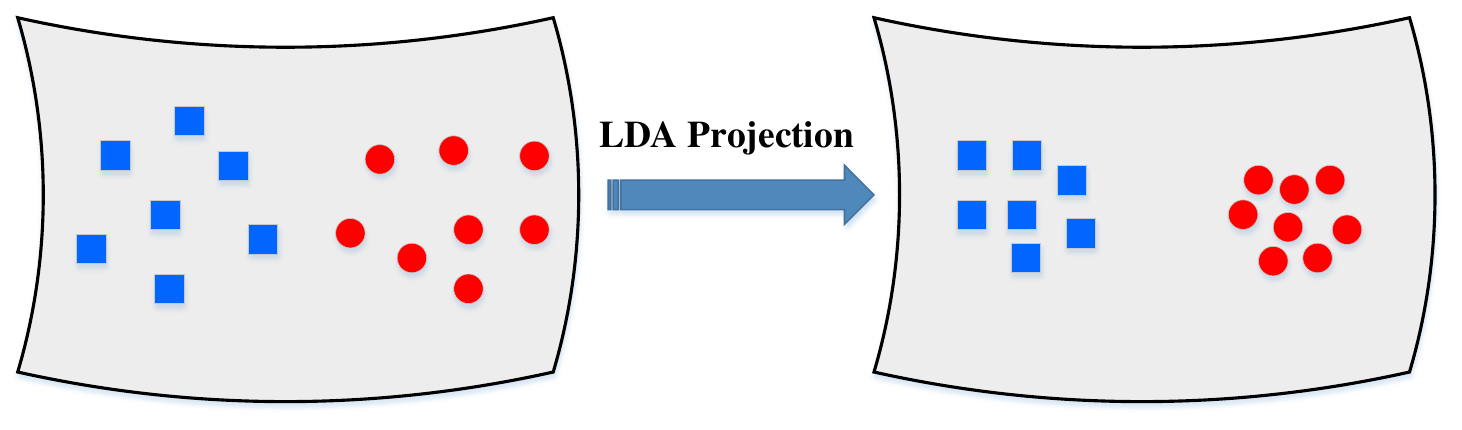}
  \caption{LDA Projection}
  \label{fig:LDA_Projection}
\end{figure}

As shown numerically using the Riemannian gradient descent algorithm in Section \ref{sec:ns} (see also  Section \ref{sec:SMP}), the matrix $\bm{W}$ can be uniquely  learnt from the optimization problem for a given training dataset. For any $\bm{W} \in \operatorname{St}(M,\mathbb{C}^N)$, if $\bm{R}\succ \bm{0}$, then $\bm{W}^{\operatorname{H}}\bm{R}\bm{W} \succ \bm{0}$. Thus, signal detection can be realized in the lower-dimensional HPD matrix manifold $\mathscr{P}(M,\mathbb{C})$ by
\begin{equation}\label{eq:detection_rule}
\begin{aligned}
d(f_{\bm{W}}(\bm{R}_\mathcal{G}),f_{\bm{W}}(\bm{R}_D)) = d(\bm{W}^{\operatorname{H}}\bm{R}_\mathcal{G}\bm{W},\bm{W}^{\operatorname{H}}\bm{R}_D\bm{W}) \mathop{\gtrless}\limits_{\mathcal{H}_0}^{\mathcal{H}_1} \gamma,
\end{aligned}
\end{equation}
where $\bm{R}_\mathcal{G}$ denotes the clutter covariance matrix estimation, $\bm{R}_D$ is the observation in the cell under test, $d(\cdot,\cdot)$ is the geometric distance between two points in the manifold $\mathscr{P}(M,\mathbb{C})$, and $\gamma$ denotes the detection threshold.

Eq. \eqref{eq:detection_rule} indicates that the statistic is actually the geometric distance $d(\cdot,\cdot)$ on $\mathscr{P}(N,\mathbb{C})$. Equipped with the same geometric measure, detection performance is hence closely related to  discriminative power of the HPD matrix manifold. Through the LDA manifold projection, HPD matrices are mapped into a more discriminative lower-dimensional HPD matrix manifold, where the within-class similarity decreases while the between-class similarity increases. Consequently, the detection performance is also improved with such a manifold projection.
The framework of LDA-MIG detectors is illustrated in Fig. \ref{fig:Detection_framework}.
\begin{figure*}[!t]
  \centering
  \includegraphics[width=15cm]{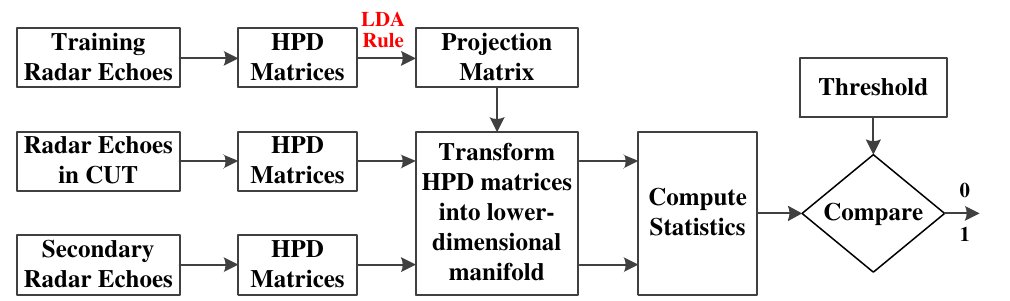}
  \caption{The framework of LDA-MIG detectors}
  \label{fig:Detection_framework}
\end{figure*}

\section{PRELIMINARIES OF MATRIX INFORMATION GEOMETRY}
\label{sec:MIG}
We briefly review the geometric theory of HPD manifolds in this section. Let $GL(N,\mathbb{F})$ be the general linear group with $\mathbb{F}$ either $\mathbb{C}$ or $\mathbb{R}$, where each element of $GL(N,\mathbb{F})$ is an $N\times N$ invertible matrix. The  the Frobenius inner product (or the Hilbert--Schmidt inner product) is defined by
\begin{equation}\label{eq:glmetric}
\langle \bm{X},\bm{Y}\rangle:=\operatorname{tr}(\bm{X}^{\operatorname{H}}\bm{Y}).
\end{equation}
We will mainly be focused on HPD matrices from now on.

For the HPD matrix manifold
\begin{equation*}
\begin{aligned}
\mathscr{P}(N,\mathbb{C}):=\left\{  \bm{Y}\in GL(N,\mathbb{C}) \mid \bm{y}^{\operatorname{H}}\bm{Y}\bm{y}>0, \forall \bm{y}\in\mathbb{C}^N/ \{0\}\right\},
\end{aligned}
\end{equation*}
an AIRM at a point $\bm{P}\in\mathscr{P}(N,\mathbb{C})$ is defined by (e.g., \cite{pennec2006,2006averaging})
\begin{equation}\label{eq:HPDmetric}
g_{\bm{P}}(\bm{A},\bm{B}):=\operatorname{tr}\left(\bm{P}^{-1}\bm{A}\bm{P}^{-1}\bm{B}\right),
\end{equation}
where $\bm{A},\bm{B}\in T_{\bm{P}}\mathscr{P}(N,\mathbb{C})$. 
It gives the geodesic distance of two HPD matrices.

\begin{defn}
The AIRM distance between two HPD matrices $\bm{X}, \bm{Y}\in \mathscr{P}(N,\mathbb{C})$ is defined as
\begin{equation}
\begin{aligned}
d_A^2(\bm{X},\bm{Y}) &= \norm{\operatorname{Log}\left(\bm{X}^{-\frac12}\bm{Y}\bm{X}^{-\frac12}\right)}^2 \\
&= \norm{\operatorname{Log}\left(\bm{X}^{-1}\bm{Y}\right)}^2,
\end{aligned}
\end{equation}
where $\operatorname{Log}$ denotes the unique  principle logarithm of an invertible matrix that does not
have eigenvalues in the closed negative real line.
\end{defn}

Beside the AIRM distance, other distance or distance-like functions can also defined on $\mathscr{P}(N,\mathbb{C})$ to measure the difference between two HPD matrices. Some of them that will be discussed in this paper are listed below.
\begin{defn}
The LEM distance, the JBLD and SKLD of two HPD matrices $\bm{X}, \bm{Y}\in \mathscr{P}(N,\mathbb{C})$ are respectively defined by
\begin{equation}
\begin{aligned}
d_L^2(\bm{X},\bm{Y}) &= \norm{\operatorname{Log}\bm{X}-\operatorname{Log}\bm{Y}}^2,\\
d_J^2(\bm{X},\bm{Y}) &= \ln\det \left(\frac{\bm{X}+\bm{Y}}{2}\right) - \frac{1}{2}\ln\det (\bm{X}\bm{Y}),\\
d_S^2(\bm{X},\bm{Y}) &= \operatorname{tr}\left( \bm{Y}^{-1}\bm{X} + \bm{X}^{-1}\bm{Y} - 2\bm{I}\right).
\end{aligned}
\end{equation}
\end{defn}

\section{GEOMETRIC MEANS AND LDA MANIFOLD PROJECTION}
\label{sec:SMP}
\subsection{Geometric Means}
Given $K$ positive numbers $\{x_k \}_{k\in[K]}$, it is well known that its arithmetic mean is
\begin{equation}
\bar{x} = \frac{1}{K}\sum_{k=1}^K x_k.
\end{equation} In fact, the arithmetic mean is the minimum value of the sum of square distances, namely
\begin{equation}
\bar{x} := \underset{x \in \mathbb{R}^+}{{\arg\min}} \frac{1}{K}\sum_{k=1}^K | x - x_k |^2,
\end{equation}
where $| x - x_k |$ is the distance between $x$ and $x_k$. Similarly, geometric mean for $K$ HPD matrices can be defined as follows.
\begin{defn}
The geometric mean, related to a geometric measure $d: \mathscr{P}(N,C) \times \mathscr{P}(N,C) \rightarrow \mathbb{R}$, for $K$ HPD matrices $\{ \bm{R}_k \}_{k\in[K]}$, is the unique solution of the following function:
\begin{equation}
\bm{\overline{R}} := \underset{\bm{R} \in \mathscr{P}(N,\mathbb{C})}{{\arg\min}} \frac{1}{K}\sum_{k=1}^K d^2(\bm{R},\bm{R}_k).
\label{eq:def_gm}
\end{equation}
\end{defn}

To find minimizer of the function
\begin{equation}
F(\bm{R})= \frac{1}{K}\sum_{k=1}^K d^2(\bm{R},\bm{{R}}_k),\quad \bm{R} \in \mathscr{P}(N,\mathbb{C}),
\end{equation}
it suffices to solve the matrix equation $\nabla F(\bm{R})=\bm{0}$, which is defined using the Riemannian or Frobenius metric of $\mathscr{P}(N,\mathbb{C})$ by the  covariant/directional derivative (see, e.g., \cite{huaetal2020})
\begin{equation}\label{def:gra}
\langle \nabla F(\bm{R}),\bm{X}\rangle:=\frac{\operatorname{d}}{\operatorname{d}\!\varepsilon}\Big|_{\varepsilon=0} F(\bm{R}+\varepsilon\bm{X}),  \forall \bm{X}\in T_{\bm{R}}\mathscr{P}(N,\mathbb{C}).
\end{equation}

The equation $\nabla F(\bm{R})=\bm{0}$ can sometimes be solved analytically to obtain geometric means of $K$ HPD matrices $\{ \bm{R}_k \}_{k\in[K]}$. Otherwise, geometric means can be derived  through the fixed-point algorithm. In the below, we summarize geometric means with respect to the four geometric measures introduced above, i.e., the AIRM distance, the LEM distance, the JBLD, and the SKLD.


\begin{prop}
The AIRM mean for $K$ HPD matrices $\{ \bm{R}_k \}_{k\in[K]}$ can be calculated through the fixed-point algorithm \cite{2006averaging}
\begin{equation}\label{eq:AIRM_Mean}
\bm{\overline{R}}_{l+1} = \bm{\overline{R}}_l^{1/2} \exp\left\{\eta_l \sum_{k=1}^K \operatorname{Log}\left(\bm{\overline{R}}_l^{-1/2}\bm{R}_k\bm{\overline{R}}_l^{-1/2}\right)\right\}\bm{\overline{R}}_l^{1/2},
\end{equation}
where the subscript $l$ denotes the iterative index, $\exp$ denotes the matrix exponential, and $\frac{K-1}{K}<\eta_l<1$ is the step size that can vary with respect to $l$. 
\end{prop}

\begin{prop}\label{prop:LEMm}
The LEM mean for $K$ HPD matrices $\{ \bm{R}_k \}_{k\in[K]}$ is \cite{AFPA2007}
\begin{equation}\label{eq:LEM_Mean}
\bm{\overline{R}} =\exp \left(\frac{1}{K}\sum_{k=1}^K \operatorname{Log} \bm{R}_k \right).
\end{equation}
\end{prop}


\begin{prop}
The JBLD mean for $K$ HPD matrices $\{ \bm{R}_k \}_{k\in[K]}$ can be calculated by
\begin{equation}\label{eq:JBLD_Mean}
\bm{\overline{R}}_{l+1} = \left(\frac{1}{K}\sum_{k=1}^K  \left(\frac{\bm{\overline{R}}_l + \bm{R}_k}{2}\right)^{-1} \right)^{-1},
\end{equation}
where the subscript $l$ denotes the iterative index again.
\end{prop}

\begin{proof}
Denote $F(\bm{R})$ as the objective function to be minimized, namely
\begin{equation}
F(\bm{R}) = \frac{1}{K}\sum_{k=1}^K \left\{\ln\det\left(\frac{\bm{R}+\bm{R}_k}{2}\right) - \frac{1}{2}\ln\det(\bm{R}\bm{R}_k)\right\}.
\end{equation}
Its gradient  can be immediately calculated via the definition \eqref{def:gra} and reads
\begin{equation}
\nabla F(\bm{R}) = \frac{1}{2K}\sum_{k=1}^K \left\{\left(\frac{\bm{R}+\bm{R}_k}{2}\right)^{-1} - \bm{R}^{-1} \right\}.
\end{equation}

Moving the $\bm{R}^{-1}$ terms in $\nabla F(\bm{R}) = \bm{0}$ to the right side, we  obtain the fixed-point algorithm.
\end{proof}

\begin{prop}\label{prop:SKLDm}
The SKLD mean for $K$ HPD matrices $\{ \bm{R}_k \}_{k\in[K]}$ is given by \cite{9764734}
\begin{equation}\label{eq:SKL_Mean}
\bm{\overline{R}} = \bm{A}^{-1/2}\left( \bm{A}^{1/2}\bm{B}\bm{A}^{1/2} \right)^{1/2}\bm{A}^{-1/2},
\end{equation}
where
\begin{equation}
\bm{A} = \sum_{k=1}^K \bm{R}_k^{-1}, \quad \bm{B} = \sum_{k=1}^K \bm{R}_k.
\end{equation}
\end{prop}


\subsection{Influence Functions}

In order to analyze the robustness of a geometric mean, we define an influence function with respect to the geometric mean for the HPD matrices contaminated by outliers, where the influence function reflects the effect of outliers on the estimate accuracy of the geometric mean.
Let $\bm{\overline{R}}$ be the geometric mean of $K$ HPD matrices $\{ \bm{R}_k \}_{k\in[K]}$, and let $\bm{\widehat{R}}$ be the geometric mean of $K$ HPD matrices and $J$ HPD outliers $\{ \bm{P}_j \}_{j\in[J]}$, with a weight $\varepsilon$ $(\varepsilon \ll 1)$ into these $K$ HPD matrices. The mean $\widehat{\bm{R}}$ is assumed to satisfy $\nabla G(\widehat{\bm{R}})=\bm{0}$ where
\begin{equation}
\begin{aligned}
G(\bm{R}) := (1-\varepsilon)\frac{1}{K}\sum_{k=1}^{K}d^2(\bm{R},\bm{R}_k) + \varepsilon\frac{1}{J}\sum_{j=1}^{J}d^2(\bm{R},\bm{P}_j).
\end{aligned}
\end{equation}
Namely, it is a minimizer of the function $G(\bm{R})$.
Furthermore, the mean $\bm{\widehat{R}}$ can be defined as a perturbation
\begin{equation}
\begin{aligned}
\bm{\widehat{R}} = \bm{\overline{R}} + \varepsilon \bm{H}\left(\{ \bm{R}_k \}_{k\in[K]},\{ \bm{P}_j \}_{j\in[J]}\right)+O(\varepsilon^2),
\end{aligned}
\end{equation}
and we define the {\em influence function} as
\begin{equation}
\begin{aligned}
&f\left(\{ \bm{R}_k \}_{k\in[K]},\{ \bm{P}_j \}_{j\in[J]}\right)  \\
&\quad\quad\quad := \norm{ \bm{H}\left(\{ \bm{R}_k \}_{k\in[K]},\{ \bm{P}_j \}_{j\in[J]}\right)},
\end{aligned}
\end{equation}
where $ \bm{H}\left(\{ \bm{R}_k \}_{k\in[K]},\{ \bm{P}_j \}_{j\in[J]}\right)$ is a Hermitian matrix depending on $\{ \bm{R}_k \}_{k\in[K]}$ and $\{ \bm{P}_j \}_{j\in[J]}$.
The arguments will often be omitted for convenience. 
Considering $K$ HPD matrices $\{ \bm{R}_k \}_{k\in[K]}$ and $J$ outliers $\{ \bm{P}_j \}_{j\in[J]}$, the influence functions with respect to the four geometric means are given below.

\begin{prop}\label{prop:AIRMinf_fun}
The influence function with respect to the AIRM mean is $h=\norm{\bm{H}}$ where
\begin{equation}\label{eq:AIRMInfluFun}
\bm{H} = - \frac{1}{J}\sum_{j=1}^{J}\frac{\bm{\overline{R}}  \operatorname{Log}\left(\bm{P}_j^{-1}\bm{\overline{R}}\right)+ \operatorname{Log}\left(\bm{\overline{R}}\bm{P}_j^{-1}\right)\bm{\overline{R}}}{2}.
\end{equation}
\end{prop}

\begin{proof}
The proof is similar to that given in \cite{huaetal2020}, where the Hermitian property of $\bm{H}$ should be taken into consideration.
\end{proof}

\begin{prop}\label{prop:LEMinf_fun}
The influence function of the LEM mean is $h=\norm{\bm{H}}$ where
\begin{equation}\label{eq:LEMInfluFun}
\bm{H} = \bm{\overline{R}}^{1/2}\left( \frac{1}{J}\sum_{j=1}^{J}\operatorname{Log}\bm{P}_j - \operatorname{Log}\bm{\overline{R}}\right)\overline{\bm{R}}^{1/2}.
\end{equation}
\end{prop}

\begin{proof}
 See Appendix \ref{app:LEMH}.
\end{proof}

\begin{prop}\label{prop:JBLDinf_fun}
The influence function of the JBLD mean  is $h=\norm{\bm{H}}$ where
\begin{equation}\label{eq:JBLDInfluFun}
\begin{aligned}
\bm{H}&= \frac{K}{J}\sum_{j=1}^{J}\left(\bm{\overline{R}}^{-1} - \left(\frac{\bm{\overline{R}}+\bm{P}_j}{2}\right)^{-1}\right)\\
&\quad\quad\quad \times \left( \sum_{k=1}^{K}\left(\bm{\overline{R}}^{-2} - \frac{1}{2}\left(\frac{\bm{\overline{R}}+\bm{R}_k}{2}\right)^{-2}\right) \right)^{-1}.
\end{aligned}
\end{equation}
\end{prop}

\begin{proof}
See Appendix \ref{app:JBLDH}.
\end{proof}

\begin{prop}\label{prop:SKLinf_fun}
The influence function of the SKLD mean is $h=\norm{\bm{H}}$ in which we have
\begin{equation}\label{eq:SKLInfluFun}
\begin{aligned}
\bm{H}&=\frac{K}{J}\left(\left(\sum_{k=1}^K\bm{R}_k^{-1}\right)\overline{\bm{R}}+\overline{\bm{R}}\left(\sum_{k=1}^K\bm{R}_k^{-1}\right)\right)^{-1}\\
& \quad \quad \quad \times \sum_{j=1}^J\left(\bm{P}_j-\overline{\bm{R}}\bm{P}_j^{-1}\overline{\bm{R}}\right).
\end{aligned}
\end{equation}
\end{prop}

\begin{proof}
See Appendix \ref{app:SKLDH}.
\end{proof}

\subsection{LDA Manifold Projection}
\label{sec:smp1}
In this subsection, we describe the LDA manifold projection for deriving the projection matrix $\bm{W}$. The LDA manifold projection maps HPD matrices from a higher-dimensional HPD matrix manifold to a more discriminative lower-dimensional one by using the sample data with class labels. Labels $0$ and $1$ denote the absence and existence of a target signal, respectively. Specifically, we use training $N\times N$ HPD matrices $\{ \bm{X}_i, x_i \}_{i\in[m]}$ and $\{ \bm{Y}_j, y_j \}_{j\in[n]}$ to perform the supervised learning for the LDA manifold projection, where $\{ \bm{X}_i \}_{i\in[m]}$ denotes the set of HPD matrices calculated by the sample data containing a target signal according to \eqref{eq:AR1} with the class label $\{x_i = 1\}_{i\in[m]}$, and $\{ \bm{Y}_j \}_{j\in[n]}$ represents the set of clutter covariance matrix estimated as the geometric mean of the secondary HPD matrices constructed by the sample data containing only the clutter with the label $\{y_j = 0\}_{j\in[n]}$. The projection is defined by considering the cost function which  resorts  to the between-class and within-class distances as follows:
\begin{equation}\label{eq:psifW}
\psi(\bm{W}) = \varphi_w(f_{\bm{W}}(\bm{X});f_{\bm{W}}(\bm{Y})) - \varphi_b(f_{\bm{W}}(\bm{X});f_{\bm{W}}(\bm{Y})),
\end{equation}
where $\varphi_w$ denotes the within-class distance, $\varphi_b$ denotes  the between-class distance, and for simplicity, we focus on the linear function $f_{\bm{W}}:\mathscr{P}(N,\mathbb{C}) \rightarrow \mathscr{P}(M,\mathbb{C})$ ($N>M$) defined by
\begin{equation}\label{eq:fW}
f_{\bm{W}}(\bm{R})=\bm{W}^{\operatorname{H}}\bm{R}\bm{W}.
\end{equation}
The matrix $\bm{W}\in\mathbb{C}^{N\times M}$ is of maximal rank\footnote{
 For any $\bm{R}\in\mathscr{P}(N,\mathbb{C})$, it can be shown directly that the conjugate symmetric matrix $\bm{W}^{\operatorname{H}}\bm{R}\bm{W}$ is positive definite if  that for any $\bm{x}\in\mathbb{C}^M$, $\bm{W}\bm{x}=\bm{0}$ if and only if $\bm{x}=0$, namely $\operatorname{rank} \bm{W}=M$.}, i.e., $\operatorname{rank} \bm{W}=M$. To maintain spectral information of the matrix $\bm{R}$ as much as possible, we impose a further constraint that
 \begin{equation}
 \bm{W}^{\operatorname{H}}\bm{W}=\bm{I}_M.
 \end{equation}
 Consequently, the matrix $\bm{W}$  is in the Stiefel manifold \eqref{Stmanifold}.
The Stiefel manifold $\operatorname{St}(M,\mathbb{C}^N)$ with $N\geq M$ is a $\left(2NM- M(M+1)\right)$-dimensional embedded submanifold of $\mathbb{C}^{N\times M}$. It is closed, bounded and compact with respect to the induced Frobenius inner product\cite{AMS2008}.

 To improve the discriminative power in the sense of  maximizing the between-class distance while minimizing the within-class distance on the projection manifold 
,  the cost function $\psi(\bm{W})$ will be minimized. Then, learning the manifold mapping can be formulated as an optimization problem on a Stiefel manifold, namely
\begin{equation}\label{eq:opt_problem}
\begin{aligned}
\underset{\bm{W}\in \operatorname{St}(M,\mathbb{C}^N)}{\operatorname{min}} \psi(\bm{W}),
\end{aligned}
\end{equation}
where $\psi(\bm{W})$ is given by Eq. \eqref{eq:psifW}.


For simplicity, we only take the neighbourhood of each HPD matrix $\bm{X}_i\in \mathscr{P}(N,\mathbb{C})$ into account. Denote $\chi^w(\bm{X}_i)$  the set of $\nu^w_i$ number of close neighbours of $\bm{X}_i$ sharing the same labels, and $\chi^b(\bm{X}_i)$  the set of $\nu^b_i$ number of close neighbours of $\bm{X}_i$ with  different labels. In practice, elements of  $\chi^w(\bm{X}_i)$ and $\chi^b(\bm{X}_i)$  can be  chosen as those in closed balls of $ \mathscr{P}(N,\mathbb{C})$ under the Frobenius distance, that are centered at $\bm{X}_i$ and with a proper radius; alternatively, one may simply fix the numbers $\nu^w_i$ and $\nu^b_i$. Further details on how to choose a neighbourhood in simulations are provided in Section \ref{sec:ns}. As a consequence, the functions of within-class distance and between-class distance are respectively given by
\begin{equation*}
\begin{aligned}
\varphi_w(f_{\bm{W}}(\bm{X});f_{\bm{W}}(\bm{Y})) &= \sum_{i=1}^m \sum_{j=1}^{\nu^w_i} d^2(f_{\bm{W}}(\bm{X}_i),f_{\bm{W}}(\bm{X}_j)) \\
&~~~+ \sum_{i=1}^n \sum_{j=1}^{\nu^w_i} d^2(f_{\bm{W}}(\bm{Y}_i),f_{\bm{W}}(\bm{Y}_j))
\end{aligned}
\end{equation*}
and
\begin{equation*}
\begin{aligned}
\varphi_b(f_{\bm{W}}(\bm{X});f_{\bm{W}}(\bm{Y})) &= \sum_{i=1}^m \sum_{j=1}^{\nu^b_i} d^2(f_{\bm{W}}(\bm{X}_i),f_{\bm{W}}(\bm{Y}_j))\\
&~~~+ \sum_{i=1}^n \sum_{j=1}^{\nu^b_i} d^2(f_{\bm{W}}(\bm{Y}_i),f_{\bm{W}}(\bm{X}_j)),
\end{aligned}
\end{equation*}
where $f_{\bm{W}}$ is defined by Eq. \eqref{eq:fW} with $\bm{W}\in  \operatorname{St}(M,\mathbb{C}^N)$ and $d^2(\cdot,\cdot)$ denotes a geometric measure of two HPD matrices. In the next, we will specify it by using the  AIRM distance, the LEM distance, the JBLD, and the SKLD, respectively.

In this paper, we will solve the optimization problem \eqref{eq:opt_problem} by using the Riemannian gradient descent algorithm \cite{Smith1993,Udr1994}; see Appendix \ref{app:SM}. The Riemannian gradient of  a function $\psi(\bm{W})$ defined on the Stiefel manifold $\operatorname{St}(M,\mathbb{C}^N)$ is given by \cite{AMS2008}
\begin{equation}\label{eq:Rgrad}
\operatorname{grad}\psi(\bm{W}) = \nabla \psi(\bm{W})-\bm{W}\times \operatorname{sym}\left(\bm{W}^{\operatorname{H}}\nabla \psi(\bm{W})\right),
\end{equation}
where $\operatorname{sym}(\bm{A})=(\bm{A}+\bm{A}^{\operatorname{H}})/2$
denotes the symmetric part of a matrix $\bm{A}$, and $\nabla\psi(\bm{W})$ is the Euclidean gradient induced from the Frobenius inner product (see \eqref{def:gra}).

To compile the Riemannian gradient descent algorithms corresponding to the geometric measures of our interest,  we must study their Riemannian gradients respectively. It suffices to compute the Euclidean gradient of the function $d^2(f_{\bm{W}}(\bm{X}),f_{\bm{W}}(\bm{Y}))$ about its argument $\bm{W}\in \operatorname{St}(M,\mathbb{C}^N)$ for given (but arbitrary) $\bm{X},\bm{Y}\in\mathscr{P}(N,\mathbb{C})$, that will allow us to express the Euclidean gradient $\nabla \psi(\bm{W})$ simply as a linear combination; finally, the Riemannian gradient $\operatorname{grad} \psi(\bm{W})$ will be obtained through Eq. \eqref{eq:Rgrad}. In the bellow, we will summarize the  Euclidean gradients of the functions $d^2(f_{\bm{W}}(\bm{X}),f_{\bm{W}}(\bm{Y}))$ about $\bm{W}$, with $f_{\bm{W}}$ given by Eq. \eqref{eq:fW}, and $d$ the AIRM distance, the LEM distance, the JBLD and the SKLD, respectively.

\begin{prop}\label{prop:AIRMEg}
The Euclidean gradient of the AIRM distance $d^2_A(f_{\bm{W}}(\bm{X}),f_{\bm{W}}(\bm{Y}))$ about $\bm{W}$ is given by
\begin{equation}
\begin{aligned}
&\frac14 \nabla d_A^2(f_{\bm{W}}(\bm{X}),f_{\bm{W}}(\bm{Y})) \\
&~~~~ =  \bm{Y}\bm{W}\operatorname{Log}\left(\left(\bm{W}^{\operatorname{H}}\bm{X}\bm{W}\right)^{-1}\left(\bm{W}^{\operatorname{H}}\bm{Y}\bm{W}\right)\right)\\
&~~~~~~~~~~~~\times \left(\bm{W}^{\operatorname{H}}\bm{Y}\bm{W}\right)^{-1}  \\
&~~~~~~~~-  \bm{X}\bm{W}  \operatorname{Log} \left(\left(\bm{W}^{\operatorname{H}}\bm{X}\bm{W}\right)^{-1}\left(\bm{W}^{\operatorname{H}}\bm{Y}\bm{W}\right)\right) \\
&~~~~~~~~~~~~\ \times \left(\bm{W}^{\operatorname{H}}\bm{X}\bm{W}\right)^{-1} .
\end{aligned}
\end{equation}
\end{prop}

\begin{proof}
See Appendix \ref{app:AIRMEg}.
\end{proof}

\begin{prop}\label{prop:LEMEg}
The Euclidean gradient of the LEM distance $d^2_L(f_{\bm{W}}(\bm{X}),f_{\bm{W}}(\bm{Y}))$ about $\bm{W}$ is given by
\begin{equation}
\begin{aligned}
&\frac14 \nabla d_L^2(f_{\bm{W}}(\bm{X}),f_{\bm{W}}(\bm{Y})) \\
&~~ =  \bm{X}\bm{W}\bigg(\left(\bm{W}^{\operatorname{H}}\bm{X}\bm{W}\right)^{-1}\operatorname{Log}\left(\bm{W}^{\operatorname{H}}\bm{X}\bm{W}\right)\\
&~~~~ - \int_0^1\left[\left(\bm{W}^{\operatorname{H}}\bm{X}\bm{W}-\bm{I}\right)s+\bm{I}\right]^{-1}\operatorname{Log}\left(\bm{W}^{\operatorname{H}}\bm{Y}\bm{W}\right)\\
&~~~~~~~~~~~~ \times \left[\left(\bm{W}^{\operatorname{H}}\bm{X}\bm{W}-\bm{I}\right)s+\bm{I}\right]^{-1} \operatorname{d}\!s  \bigg)\\
&~~~~  + \bm{Y}\bm{W}\bigg(\left(\bm{W}^{\operatorname{H}}\bm{Y}\bm{W}\right)^{-1}\operatorname{Log}\left(\bm{W}^{\operatorname{H}}\bm{Y}\bm{W}\right)\\
&~~~~ - \int_0^1\left[\left(\bm{W}^{\operatorname{H}}\bm{Y}\bm{W}-\bm{I}\right)s+\bm{I}\right]^{-1}\operatorname{Log}\left(\bm{W}^{\operatorname{H}}\bm{X}\bm{W}\right)\\
&~~~~~~~~~~~~ \times \left[\left(\bm{W}^{\operatorname{H}}\bm{Y}\bm{W}-\bm{I}\right)s+\bm{I}\right]^{-1} \operatorname{d}\!s  \bigg).
\end{aligned}
\end{equation}
\end{prop}

\begin{prop}
The Euclidean gradient of the JBLD  $d^2_J(f_{\bm{W}}(\bm{X}),f_{\bm{W}}(\bm{Y}))$ about $\bm{W}$ is given by
\begin{equation}
\begin{aligned}
& \nabla d_J^2(f_{\bm{W}}(\bm{X}),f_{\bm{W}}(\bm{Y})) \\
&~~~= 2\left(\bm{X}+\bm{Y}\right)\bm{W}\left(\bm{W}^{\operatorname{H}}\bm{X}\bm{W}+\bm{W}^{\operatorname{H}}\bm{Y}\bm{W}\right)^{-1}\\
&~~~~~~~ - \bm{X}\bm{W}\left(\bm{W}^{\operatorname{H}}\bm{X}\bm{W}\right)^{-1}-\bm{Y}\bm{W}\left(\bm{W}^{\operatorname{H}}\bm{Y}\bm{W}\right)^{-1}.
\end{aligned}
\end{equation}
\end{prop}

\begin{proof}
It can be verified directly from the definition \eqref{def:gra} by taking the following identity into account that
\begin{equation}
\frac{\operatorname{d}}{\operatorname{d}\!\varepsilon}\det \bm{A}(\varepsilon) = \det \bm{A}(\varepsilon) \operatorname{tr}\left( \bm{A}^{-1}(\varepsilon)\frac{\operatorname{d}}{\operatorname{d}\!\varepsilon}\bm{A}(\varepsilon)\right)
\end{equation}
holds for any invertible matrix $\bm{A}(\varepsilon)$.
\end{proof}

\begin{prop}
The Euclidean gradient of the SKLD  $d^2_S(f_{\bm{W}}(\bm{X}),f_{\bm{W}}(\bm{Y}))$ about $\bm{W}$ is given by
\begin{equation*}
\begin{aligned}
&\frac12 \nabla d_S^2(f_{\bm{W}}(\bm{X}),f_{\bm{W}}(\bm{Y})) =  \bm{X}\bm{W}\bigg[\left(\bm{W}^{\operatorname{H}}\bm{Y}\bm{W}\right)^{-1}  \\
&~~~~~~ -  \left(\bm{W}^{\operatorname{H}}\bm{X}\bm{W}\right)^{-1}   \left(\bm{W}^{\operatorname{H}}\bm{Y}\bm{W}\right)  \left(\bm{W}^{\operatorname{H}}\bm{X}\bm{W}\right)^{-1}   \bigg]\\
&~~+\bm{Y}\bm{W}\bigg[\left(\bm{W}^{\operatorname{H}}\bm{X}\bm{W}\right)^{-1} \\
&~~~~~~ -\left(\bm{W}^{\operatorname{H}}\bm{Y}\bm{W}\right)^{-1}\left(\bm{W}^{\operatorname{H}}\bm{X}\bm{W}\right)\left(\bm{W}^{\operatorname{H}}\bm{Y}\bm{W}\right)^{-1}\bigg].
\end{aligned}
\end{equation*}
\end{prop}

\begin{proof}
It can be straightforwardly computed from the definition \eqref{def:gra}. Details are omitted.
\end{proof}

\section{PERFORMANCE ASSESSMENT}
\label{sec:ns}

This section is dedicated to investigating the performance of the proposed LDA-MIG detectors in terms of their robustness to outliers as well as the probability of detection $P_d$. For comparison purposes, we also show the performance of MIG detectors without the LDA, the adaptive matched filter (AMF), the moving target detection (MTD), the adaptive coherence estimator (ACE) \cite{1381736}, the weighted amplitude iteration (WAI)-based constant false alarm rate (CFAR) \cite{7859306}. The analysis is conducted resorting to both simulated and real radar data.

To assess the robustness and $P_d$, $2000$ independent trials are repeated by the standard Monte Carlo counting techniques. The detection threshold is estimated to ensure a preassigned probability of false alarm $P_{fa}$ by using the $100/P_{fa}$ independent trials. The sample data is generated according to an $N$-dimensional zero mean complex circular Gaussian distribution with a known covariance matrix
\begin{equation}\label{eq:samdata}
\bm{C} = \sigma_c^2\bm{C}_0 + \sigma_n^2 \bm{I},
\end{equation}
where $\sigma_c^2\bm{C}_0$ denotes the clutter component with $\sigma_c^2$ the clutter power while $\sigma_n^2\bm{I}$ is the thermal noise component with $\sigma_n^2$ the noise power. Thus, the clutter-to-noise ratio (CNR) is given by CNR$=\sigma_c^2/\sigma_n^2$. The structure of the clutter covariance matrix $\bm{C}_0$ is Gaussian shaped with one-lag correlation coefficient $\rho$, and its entry is given by $[\bm{C}_0]_{i,j} = \rho^{| i-j|} \exp(\operatorname{i}2\pi f_c (i-j)), i,j\in[N]$ with $f_c$ the clutter normalized Doppler frequency. A set of $K$ secondary HPD matrices computed by \eqref{eq:AR1} is employed to estimate the clutter covariance matrix matrix by the geometric mean $\bm{R}_\mathcal{G}$. The sample data $\bm{y}_D$ is used to compute the $\bm{R}_D$ of the CUT. In the following, we set $\sigma_n^2 = 1$, CNR=$25$ dB, $\rho = 0.95$ and $f_c = 0.1$.

For the sake of clarification, the AMF and ACE are respectively given by
\begin{equation}\label{AMF}
\begin{aligned}
\Lambda_{AMF} &= \frac{| \bm{y}_D^{\operatorname{H}}\bm{\widehat{M}}_1^{-1}\bm{s}|^2}{\bm{s}^{\operatorname{H}}\bm{\widehat{M}}_1^{-1}\bm{s}} \mathop{\gtrless}\limits_{\mathcal{H}_0}^{\mathcal{H}_1} \gamma_{AMF}, \\
\Lambda_{ACE} &= \frac{| \bm{y}_D^{\operatorname{H}}\bm{\widehat{M}}_2^{-1}\bm{s}|^2}{(\bm{s}^{\operatorname{H}}\bm{\widehat{M}}_2^{-1}\bm{s})(\bm{y}_D^{\operatorname{H}}\bm{\widehat{M}}_2^{-1}\bm{y}_D)}\mathop{\gtrless}\limits_{\mathcal{H}_0}^{\mathcal{H}_1} \gamma_{ACE}.
\end{aligned}
\end{equation}
Here, $\gamma_{AMF}$ and $\gamma_{ACE}$ are the detection thresholds. $\bm{\widehat{M}}_1$ and $\bm{\widehat{M}}_2$ denote the SCM and normalized SCM estimators, which are derived by the maximum likelihood estimation of the secondary sample data. The SCM estimator is given by \eqref{eq:SCM}.


Similarly, the LDA-MIG detectors with respect to the geometric measures are given by
\begin{equation}
\Lambda_\mathcal{G} = d_{\mathcal{G}}^2(f_{\bm{W}}(\bm{R}_D),f_{\bm{W}}(\bm{R}_\mathcal{G})) \mathop{\gtrless}\limits_{\mathcal{H}_0}^{\mathcal{H}_1} \gamma_\mathcal{G}, \\
\end{equation}
where  the function $f_{\bm{W}}(\bm{R})$ is given by \eqref{eq:fW} with $\bm{W}$ to be learnt using the manifold projection, $\mathcal{G}$ denotes the AIRM, the LEM,  the JBLD  or the SKLD, and $\bm{R}_{AIRM}, \bm{R}_{LEM}, \bm{R}_{JBLD}$ and $\bm{R}_{SKLD}$ are the geometric means that are used as the clutter covariance matrix estimators.

\subsection{Robustness Analysis}

Firstly, we analyze the robustness of geometric means to outliers and compare them with the SCM estimator. A number of $K=50$ sample data is generated and their corresponding HPD matrices are  computed via Eq. \eqref{eq:AR1}. A number of $L$ interferences are injected into the $K$  HPD data. Similar to the clutter data, interferences are generated by a zero mean complex circular Gaussian distribution with the covariance matrix $\bm{C}_I = \bm{C} + \sigma_I^2\bm{s}\bm{s}^{\operatorname{H}}$, where $\bm{C}$ is given by Eq. \eqref{eq:samdata}, $\sigma_I^2$ is the interference power, and $\bm{s}$ is the steering vector. The interference-to-clutter ratio is defined by $\sigma_I^2/\sigma_n^2$.
The signal-to-clutter ratio (SCR) is defined by
\begin{equation}
\textrm{SCR} = |\alpha|^2 \bm{s}^{\operatorname{H}}\bm{\widehat{M}} ^{-1}\bm{s},
\end{equation}
where the normalized Doppler frequency of the interference is $f_I = 0.22$, and $\bm{\widehat{M}} $ is the SCM estimator given by \eqref{eq:SCM}. The dimension of the sample data is $N=8$. The number of interference $L$ varies from $1$ to $40$.

Values of the influence functions for the AIRM, JBLD, LEM and SKLD means are respectively given by Eqs. \eqref{eq:AIRMInfluFun}, \eqref{eq:LEMInfluFun}, \eqref{eq:JBLDInfluFun}, and \eqref{eq:SKLInfluFun}. The influence function with respect to SCM can be computed in a similar way. $2000$ times of computation is repeated and then averaged. The influence values of the AIRM, LEM, JBLD and SKLD means as well as the SCM is shown in Fig. \ref{Robust_Outliers}. It is clear that the SCM leads to much larger influence value than the geometric means, and the JBLD mean has the lowest influence value. It suggests that the SCM is much more sensitive to outliers than the geometric means and the JBLD mean is most robust. In addition, the AIRM mean has a similar robustness to the LEM mean, and both of them have better robustness than the SKLD mean.

\begin{figure}[ht]
  \centering
  \includegraphics[width = 8.5cm]{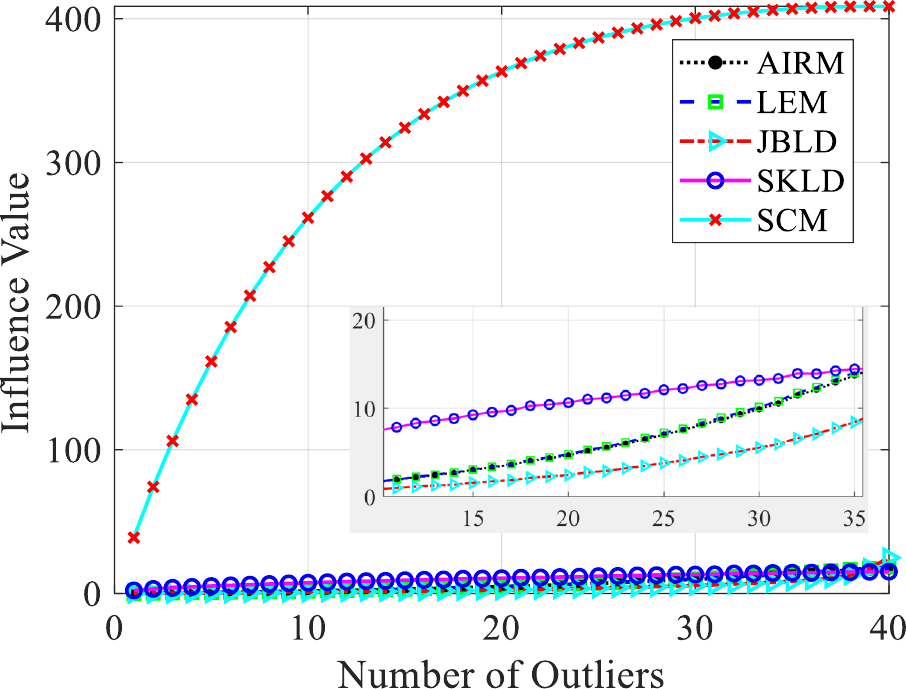}\\
  \caption{Influence values of geometric means}
  \label{Robust_Outliers}
\end{figure}

\subsection{Simulation Results}

We now examine the performance of the proposed LDA-MIG detectors with different sizes of $K$ secondary data, where $K=8$ and $16$, respectively. The dimension of the sample data is chosen to be $N=8$. The $P_{fa}$ is set to be $10^{-5}$. Two interferences are injected into the secondary data with the normalized Doppler frequency $f_I = 0.22$ and the interference power $30$ dB. The normalized Doppler frequency of target signal is given by $f_d = 0.2$.
The sample data in the CUT is generated according to an $N$-dimensional zero mean complex circular Gaussian distribution with a known covariance matrix $\bm{C}_t = \tau\bm{C} + \bm{q}\bm{q}^{\operatorname{H}}$, where $\bm{C}$ is the covariance matrix in the reference cell that is given by \eqref{eq:samdata}, $\tau$ is the nonhomogeneous power coefficient and is set to be $1.2$, and $\bm{q}$ is a random vector.
The training dataset is composed of the following two sets: the set of HPD matrices with a signal and the set of clutter covariance matrix; the size of each set is $1000$. By using the training dataset, we derive the three projection matrices that transforms the HPD matrices from the $8$-dimensional manifold to $M=6, 4, 2$, and $1$-dimensional manifolds, respectively. Furthermore, the learning is conducted with the number of neighbours sharing the same labels $\nu^w_i=15$, and the  number of neighbours  with  different labels $\nu^b_i=20$. The set of closest matrices ($15$ or $20$ of them) is selected with respect to the same distance/divergence function that is used in the detection.
Figs. \ref{K_8_N_8}, and \ref{K_16_N_16} plot the $P_d$ versus SCR for the LDA-MIG detectors and their corresponding counterparts as well as the AMF, ACE, MTD and WAI-CFAR under different size of secondary data, which are commonly used for maritime target detection in nonhomogeneous sea clutter.
Unlike the AMF that the optimal performance is the AMF with the known clutter covariance matrix, the optimal performance of MIG detectors is not the MIG detectors with the known clutter covariance matrix since the detection performance is closely related to the discrimination between the target signal and the clutter.

It is clear from Figs. \ref{K_8_N_8}, and \ref{K_16_N_16} that the detection performances of all the considered detectors are improved as $K$ becomes larger. In Fig. \ref{K_8_N_8}, the MIG detectors still work well while the $P_d$ of the AMF and ACE are very low; this is because the estimate accuracy of the SCM and normalized SCM are worse when $K=M$. It is also noted that all the LDA-MIG detectors have better performance than their corresponding MIG detectors, and the MIG detectors outperform the AMF, MTD, ACE and WAI-CFAR detectors. Specifically, for $K=M$ in simulated data, the $P_d$ of AMF and ACE are almost zero due to poor estimate accuracy of the clutter covariance matrix. However, the LDA-MIG, MIG, MTD and WAI-CFAR detectors can still work: the LDA-MIG detectors have the best performance, followed by the MIG detectors. The WAI-CFAR detector has better performance than the ACE. For $K=2M$ in simulated data, the ACE has better performance than the WAI-CFAR detector, and both perform better than the AMF and MTD. Without surprising, the LDA-MIG detectors have better performance than the MIG detectors and the conventional detectors.
The results demonstrate that the LDA manifold projection can indeed promote the discriminative power of the HPD matrices.
In more details, for the LDA-AIRM,  LDA-LEM and LDA-SLKD detectors has the best detection performance with $M=1$ and $K=M$. However, this phenomenon does not occur for the LDA-LEM detectors with $K=2M$. For the LDA-JBLD detectors, detection performance of $M=4, 2, 1$ are almost the same, and all these are better than $M=6$. A possible explanation is that the major useful information can be well preserved by the manifold projection while the redundant information can be largely cut down. To analytically determine the optimal dimension for the manifold projection is, however, very challenging, as both useful and redundant information may be reduced during the projection but their ratio is difficult to be quantified.

\begin{figure*}[!t]
\centering
\includegraphics[width=15cm,angle=0]{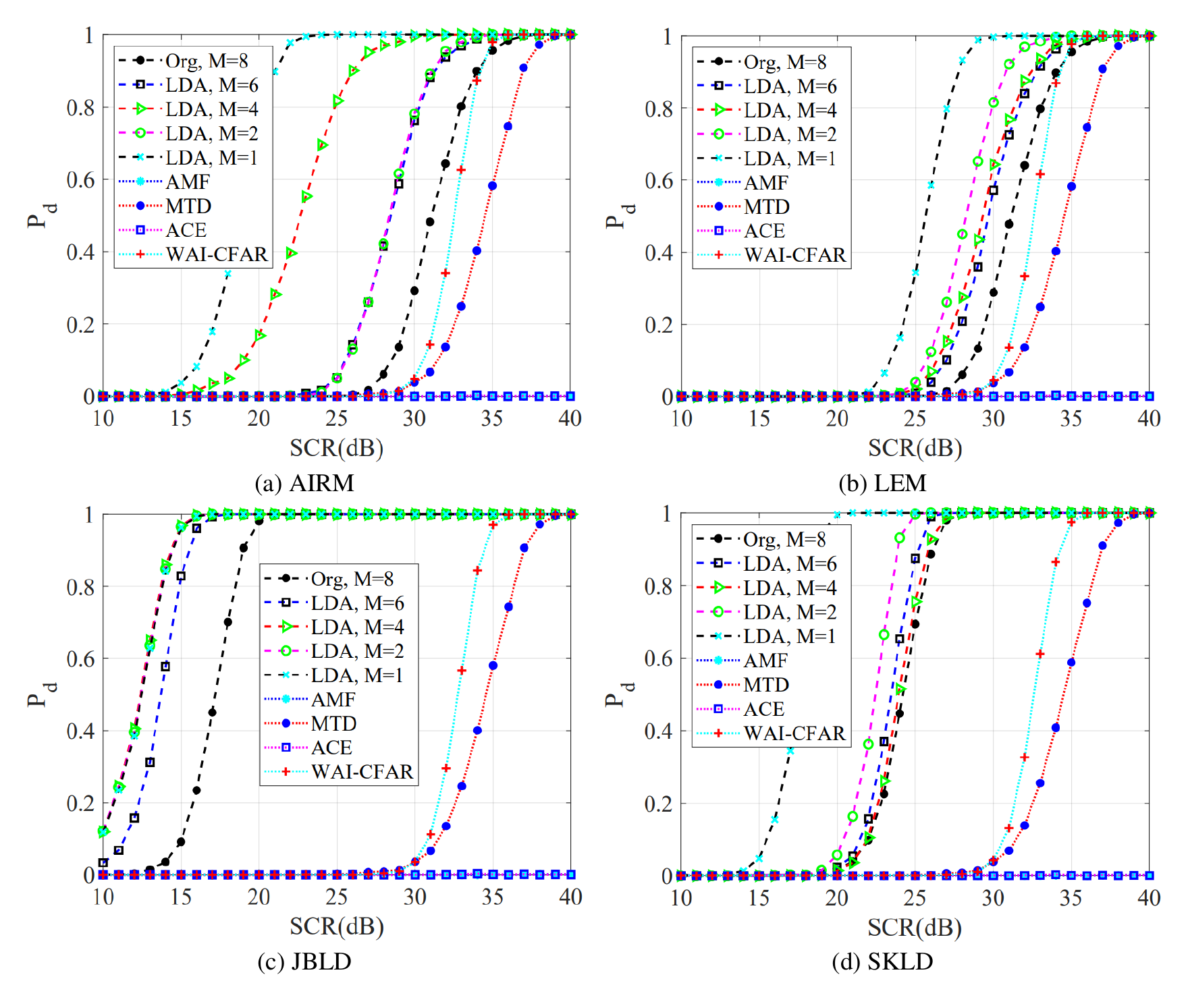}
\caption{$P_d$ vs SCR, $K=M$}
\label{K_8_N_8}
\end{figure*}

\begin{figure*}[!t]
\centering
\includegraphics[width=15cm,angle=0]{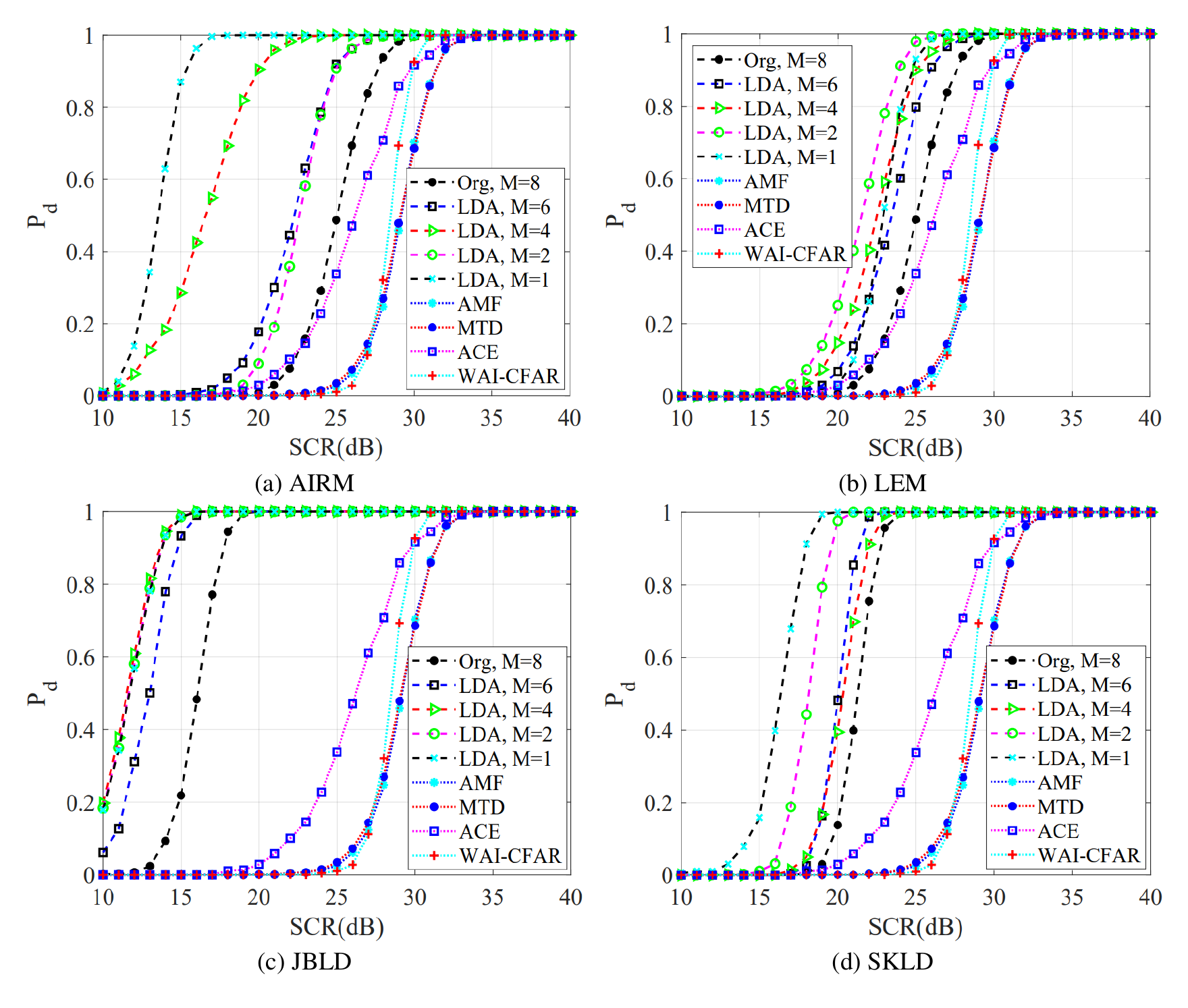}
\caption{$P_d$ vs SCR, $K=2M$}
\label{K_16_N_16}
\end{figure*}

\subsection{Real-data Results}

We now provide experiments on the real radar data to validate the advantage of the proposed detectors. The real data used in the experiments is
collected from the McMaster University IPIX radar, which is measured on the shore of Lake Ontario, between Toronto and Niagara Falls, Grimsby, Canada, in winter $1998$. We exploit the data file `$19980205\_183709\_antstep.cdf$' that only contains the sea clutter data to assess the detection performance. The power of this data is shown in Fig. \ref{Real_Data}.
This dataset consists of $28$ range cells in the fast time or range dimension, $60000$ pulses in the slow-time dimension with a pulse repetition frequency of $1000$ Hz, and refers to different polarizations; we only show the results corresponding to horizontal polarization only. Further details on the description of the dataset can be found in \cite{892695} and the references therein. Similar to the simulation part, we add a target signal in the $14$-th range cell with the target normalized Doppler frequency $f_d = 0.2$ and inject two interferences into the secondary data with the interference normalized Doppler frequency $f_I = 0.22$. The dimension of the sample data is set to $N = 5$, and the secondary data, $K = 10$ sample data around the target cell, is used to estimate the clutter covariance matrix by \eqref{eq:AR1}. We divide the dataset into $12000$ groups along the slow-time dimension. The training data is selected by the first $500$ groups with a target signal SCR $=25$ dB in the target cell. We exploit the training data to learn the projection matrix by resorting to the AIRM, the JBLD, the LEM and the SKLD. The next $10000$ groups are used to estimate the detection threshold by the Monte Carlo technique and the $P_{fa}$ is set to $10^{-4}$. The last groups are employed to estimate the $P_d$ under different SCRs.

\begin{figure*}[!t]
\centering
\includegraphics[width=8.5cm,angle=0]{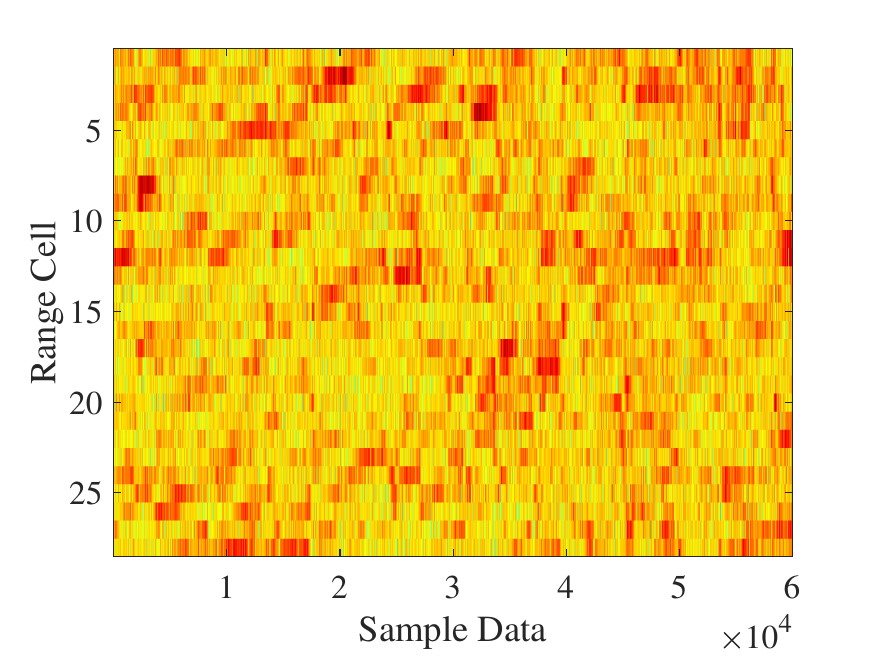}
\caption{Power of the real data.}
\label{Real_Data}
\end{figure*}

Fig. \ref{Pd_Real} shows the detection results of the considered detectors in real IPIX radar data. It is shown that the MIG detectors outperform the conventional detectors. All LDA-MIG detectors have better performance than their corresponding MIG detectors. Among the conventional detectors, the WAI-CFAR detector has the best performance, followed by the ACE, and the AMF has the worst performance.
In addition, as dimension of the reduced HPD matrix manifold decreases, the detection performance improves for the LDA-MIG detectors.
\begin{figure*}[!t]
\centering
\includegraphics[width=15cm,angle=0]{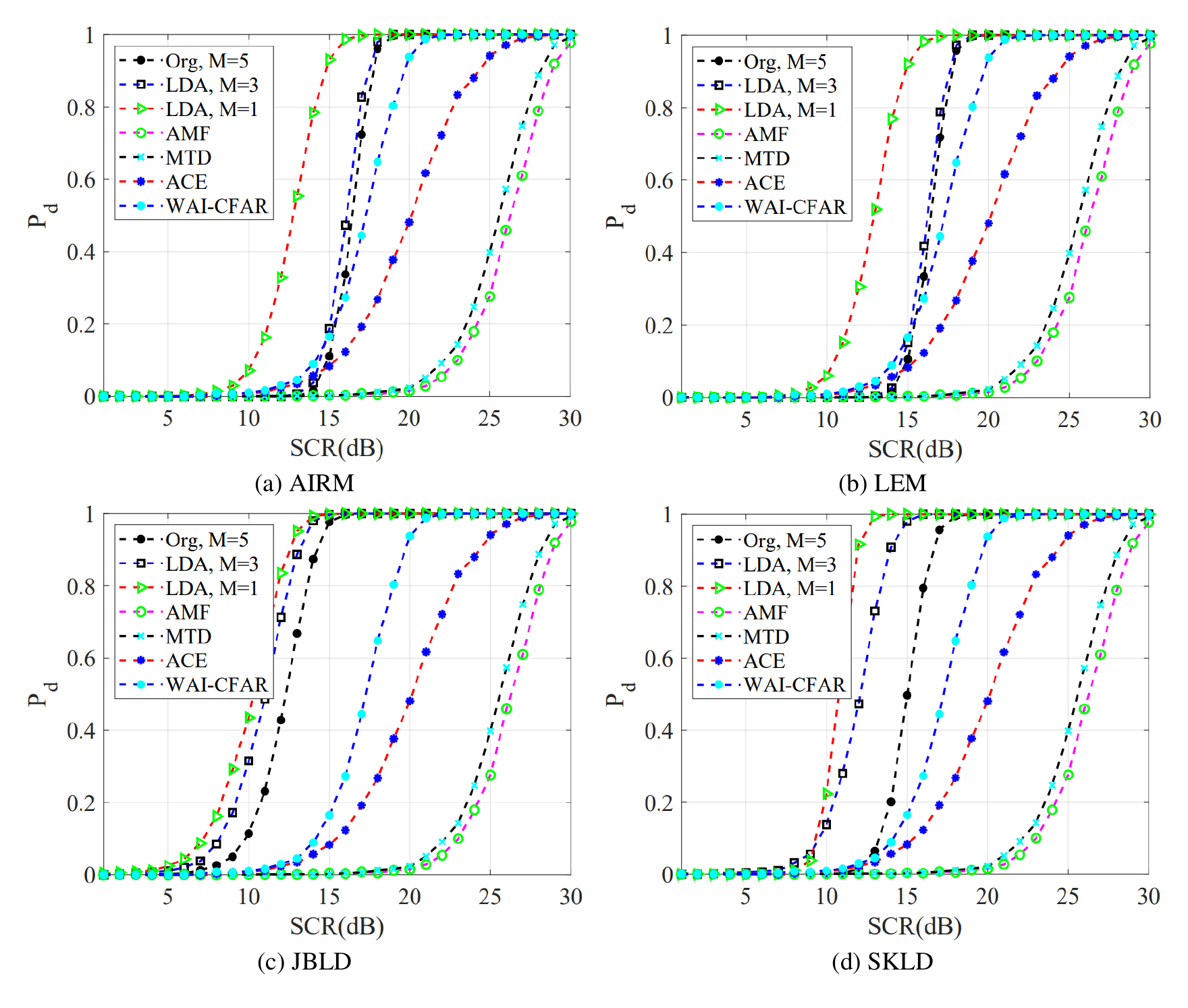}
\caption{$P_d$ versus SCR in real IPIX radar data ($K=2M$)}
\label{Pd_Real}
\end{figure*}

\subsection{Discrimination Analysis}

 To further validate the effectiveness of the LDA projection, we show the scatter plots of clutter data and clutter plus target data for SCR $= 5$ dB and SCR $= 10$ dB in Figs. \ref{Dis_SCR_5} and \ref{Dis_SCR_10}. Each sample data is modeled as an HPD matrix by Eq. \eqref{eq:AR1}. We transform each HPD matrix in a $4$-dimensional manifold by the projection matrix learnt from the training data according to the AIRM, LEM, JBLD and SKLD, respectively, which are then plotted in a $2$-dimensional space by the principal component analysis. It is noticed from Figs. \ref{Dis_SCR_5} and \ref{Dis_SCR_10} that the distance between the clutter data and clutter plus target data becomes larger after the LDA projection. Specifically, the original clutter data and clutter plus target data are mixed together and cannot be distinguished for SCR $= 5$ dB in Fig. \ref{Dis_SCR_5}. However, the data can be divided into two categories after being projected to the low dimensional manifold by the LDA projection. For SCR $= 10$ dB, the data becomes easily distinguishable after the LDA projection. These results further demonstrate the effectiveness of the proposed method.

\begin{figure*}[!t]
\centering
\includegraphics[width=18cm,angle=0]{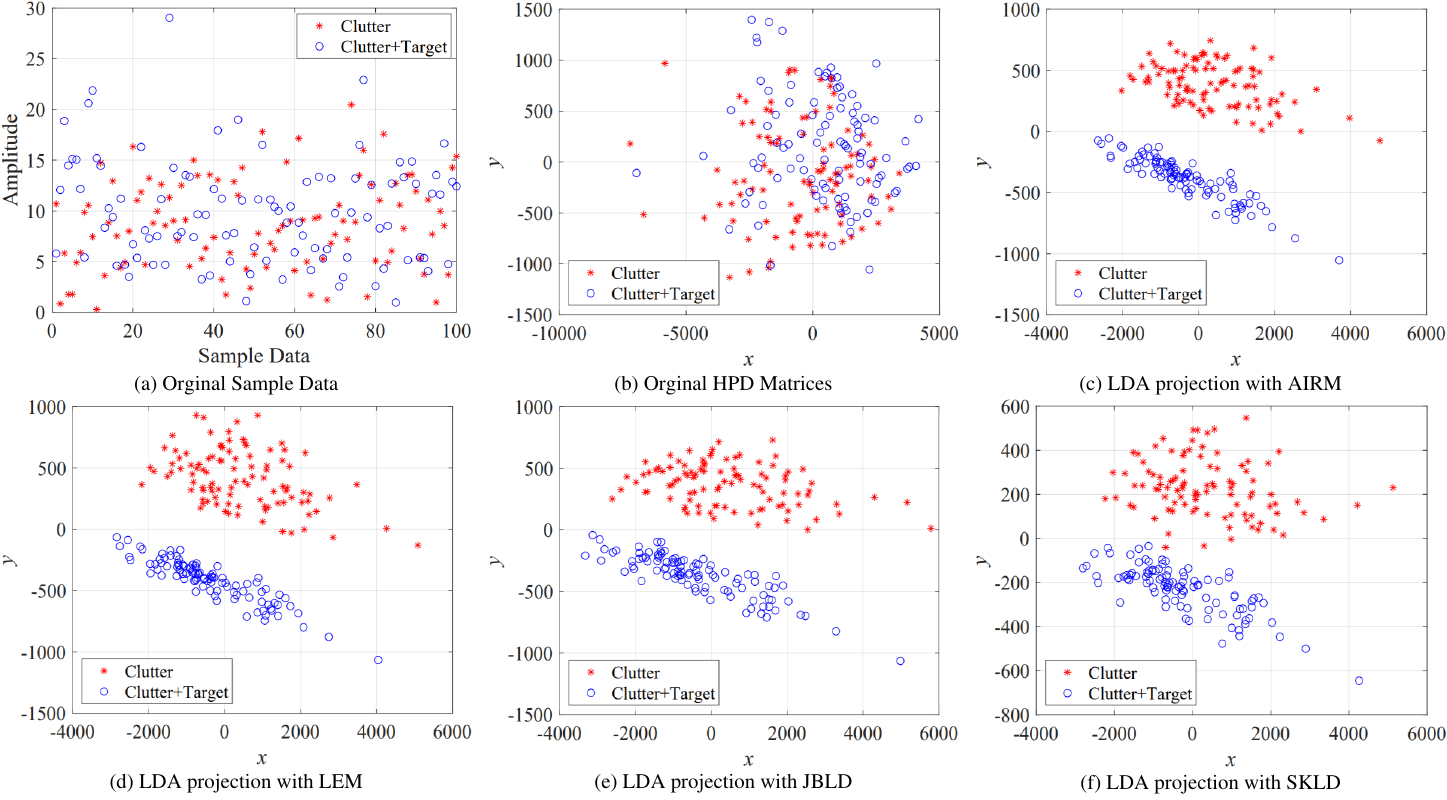}
\caption{Scatter plots with SCR $= 5$ dB.}
\label{Dis_SCR_5}
\end{figure*}

\begin{figure*}[!t]
\centering
\includegraphics[width=18cm,angle=0]{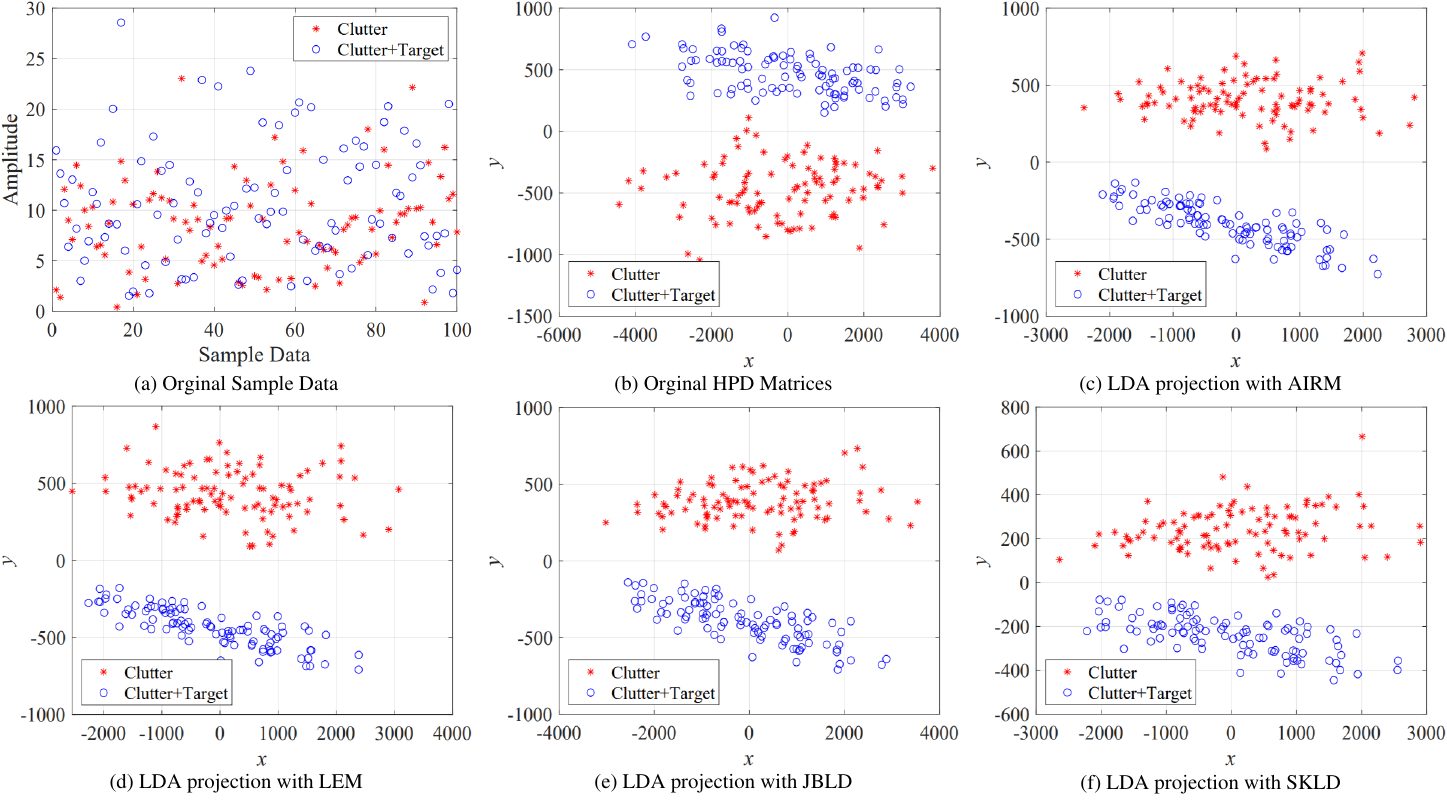}
\caption{Scatter plots with SCR $=10$ dB.}
\label{Dis_SCR_10}
\end{figure*}

\section{CONCLUSION}
\label{sec:con}

In this paper, four LDA-MIG detectors were proposed in nonhomogeneous sea clutter by incorporating a manifold projection inspired by the linear discriminant analysis. The manifold projection was defined through mapping HPD matrices from a higher-dimensional HPD manifold to a lower-dimensional one by maximizing the between-class distance while minimizing the within-class distance. We formulated such a projection as an optimization problem in the Stiefel manifold, which was solved numerically by the Riemannian gradient descent algorithm. Given a training dataset, a unique projection matrix could be  learnt for improving the discriminative power on HPD manifolds. Consequently, the LDA-MIG detectors performed better compared with respect to their counterparts without using the manifold projection as well as the AMF in nonhomogeneous clutter. However, as the projection matrix is learnt by using a training dataset, the detection performance of LDA-MIG detectors may somehow depend on the training dataset. Furthermore, we defined influence functions associated with different geometric means and calculated them in closed-form to analyze their robustness to outliers. It was shown, numerically, geometric means are more robust to outliers than the SCM.

As future researches, first of all, it is worthwhile to propose a methodology for selecting an efficiency training dataset for manifold projections. It would also be interesting to incorporate other matrix structures, e.g., shrinkage estimators and persymmetric covariance matrices, into the LDA-MIG detections, and to examine their performance differences. Finally, as illustrated through the simulations that characteristics of the data will probably affect design of the LDA-MIG detectors, further studies are needed on determining the optimal dimension of the reduced manifold for achieving the best detection performance. Finally, it would also be of interest to incorporate the accurate robustness analysis into consideration (see \cite{OP2022}).

\section*{APPENDIX}

\appendices

\section{The LEM influence function: Proof of Proposition \ref{prop:LEMinf_fun}}
\label{app:LEMH}

Consider the function
\begin{equation*}
\begin{aligned}
G(\bm{R}) &:= (1-\varepsilon)\frac{1}{K}\sum_{k=1}^{K} \norm{\operatorname{Log} \bm{R} - \operatorname{Log}\bm{R}_k}^2 \\
&\quad \quad + \varepsilon\frac{1}{J}\sum_{j=1}^{J} \norm{\operatorname{Log} \bm{R} - \operatorname{Log}\bm{P}_j}^2  .
\end{aligned}
\end{equation*}
$\nabla G(\widehat{\bm{R}})=\bm{0}$ is solved with $\widehat{\bm{R}}$ satisfying
\begin{equation*}
\begin{aligned}
&\frac{1-\varepsilon}{K}\sum_{k=1}^K\left(\operatorname{Log}\widehat{\bm{R}}-\operatorname{Log}\bm{R}_k\right)\\
&\quad\quad\quad\quad\quad +\frac{\varepsilon}{J}\sum_{j=1}^J\left(\operatorname{Log}\widehat{\bm{R}}-\operatorname{Log}\bm{P}_j\right)=\bm{0}.
\end{aligned}
\end{equation*}
We solve it to obtain
\begin{equation}\label{eq:logX}
\operatorname{Log}\widehat{\bm{R}}=\operatorname{Log}\overline{\bm{R}}+\varepsilon\left(\frac{1}{J}\sum_{j=1}^J\operatorname{Log}\bm{P}_j-\operatorname{Log}\overline{\bm{R}}\right),
\end{equation}
where $\overline{\bm{R}}$ is the LEM mean of the HPD matrices $\{ \bm{R}_k \}_{k=1}^K$, namely
\begin{equation*}
\operatorname{Log}\overline{\bm{R}}=\frac{1}{K}\sum_{k=1}^K\operatorname{Log}\bm{R}_k.
\end{equation*}
Since $\widehat{\bm{R}}=\overline{\bm{R}}+\varepsilon \bm{H}+O(\varepsilon^2)$, we differentiate Eq. \eqref{eq:logX} about $\varepsilon$, amounting to
\begin{equation*}
\begin{aligned}
&\frac{1}{J}\sum_{j=1}^J\operatorname{Log}\bm{P}_j-\operatorname{Log}\overline{\bm{R}}=\frac{\operatorname{d}}{\operatorname{d}\!\varepsilon}\Big|_{\varepsilon=0}\operatorname{Log}\widehat{\bm{R}}\\
&=\left(\int_0^1[(\widehat{\bm{R}}-\bm{I})s+\bm{I}]^{-1}\frac{\operatorname{d}\!\widehat{\bm{R}}}{\operatorname{d}\!\varepsilon}[(\widehat{\bm{R}}-\bm{I})s+\bm{I}]^{-1}\operatorname{d}\!s\right)\Big|_{\varepsilon=0}\\
\end{aligned}
\end{equation*}
\begin{equation*}
\begin{aligned}
=\int_0^1 [(\overline{\bm{R}}-\bm{I})s+\bm{I}]^{-1}\bm{H}[(\overline{\bm{R}}-\bm{I})s+\bm{I}]^{-1}\operatorname{d}\!s.
\end{aligned}
\end{equation*}
Taking trace on its both sides gives
\begin{equation*}
\operatorname{tr}(\overline{\bm{R}}^{-1}\bm{H})=\operatorname{tr}\left(\frac{1}{J}\sum_{j=1}^J\operatorname{Log}\bm{P}_j-\operatorname{Log}\overline{\bm{R}}\right).
\end{equation*}
Here, the following identity (\cite{Hig2008,Moa2005}) are applied
\begin{equation}\label{eq:logid}
\begin{aligned}
\frac{\operatorname{d}}{\operatorname{d}\!\varepsilon}\operatorname{Log}\bm{A}(\varepsilon)=&\int_0^1
[(\bm{A}(\varepsilon)-\bm{I})s+\bm{I}]^{-1}
 \frac{\operatorname{d}\!\bm{A}(\varepsilon)}{\operatorname{d}\!\varepsilon}\\
 ~~~~&\times [(\bm{A}(\varepsilon)-\bm{I})s+\bm{I}]^{-1}\operatorname{d}\!s.
\end{aligned}
\end{equation}

\begin{equation}\label{eq:intX}
\begin{aligned}
\int_0^1[(\bm{X}-\bm{I})s+\bm{I}]^{-2}\operatorname{d}\!s=\bm{X}^{-1}.
\end{aligned}
\end{equation}
Since $\overline{\bm{R}}$ is arbitrary, we can choose
\begin{equation*}
\bm{H}=\overline{\bm{R}}^{1/2}\left(\frac{1}{J}\sum_{j=1}^J\operatorname{Log}\bm{P}_j-\operatorname{Log}\overline{\bm{R}}\right)\overline{\bm{R}}^{1/2}.
\end{equation*}
This finishes the proof.

\section{The JBLD influence function: Proof of Proposition \ref{prop:JBLDinf_fun}}
\label{app:JBLDH}
Gradient of the function
 $$G(\bm{R})=(1-\varepsilon)\frac{1}{K}\sum_{k=1}^{K}d_J^2(\bm{R},\bm{R}_k) + \varepsilon\frac{1}{J}\sum_{j=1}^{J}d_J^2(\bm{R},\bm{P}_j)$$ can be computed straightforwardly:
\begin{equation*}
\begin{aligned}
\nabla G(\bm{R})&=(1-\varepsilon)\frac{1}{K}\sum_{k=1}^K\left(\left(\frac{\bm{R}+\bm{R}_k}{2}\right)^{-1}-\bm{R}^{-1}\right)\\
&\quad \quad + \frac{\varepsilon}{J}\sum_{j=1}^J\left(\left(\frac{\bm{R}+\bm{P}_j}{2}\right)^{-1}-\bm{R}^{-1}\right).
\end{aligned}
\end{equation*}
Differentiate $\nabla G(\widehat{\bm{R}})=\bm{0}$ about $\varepsilon$ at $\varepsilon=0$:
\begin{equation*}
\begin{aligned}
&\sum_{k=1}^K\left(\overline{\bm{R}}^{-1}\bm{H}\overline{\bm{R}}^{-1}-\left(\frac{\overline{\bm{R}}+\bm{R}_k}{2}\right)^{-1}\frac{\bm{H}}{2}\left(\frac{\overline{\bm{R}}+\bm{R}_k}{2}\right)^{-1}\right)\\
&\quad \quad \quad \quad \quad \quad +\frac{K}{J}\sum_{j=1}^J\left(\left(\frac{\overline{\bm{R}}+\bm{P}_j}{2}\right)^{-1}-\overline{\bm{R}}^{-1}\right)=\bm{0},
\end{aligned}
\end{equation*}
where $\overline{\bm{R}}$ satisfies
\begin{equation*}
\sum_{k=1}^K \left\{\left(\frac{\overline{\bm{R}}+\bm{R}_k}{2}\right)^{-1} - \overline{\bm{R}}^{-1} \right\}=\bm{0}.
\end{equation*}
Since we are only interested in the norm of $\bm{H}$, taking trace on both sides finishes the proof.

\section{Influence function of the SKLD mean: Proof of Proposition \ref{prop:SKLinf_fun}}
\label{app:SKLDH}

Gradient of the function
\begin{equation*}
\begin{aligned}
G(\bm{R})=(1-\varepsilon)\frac{1}{K}\sum_{k=1}^K d_S^2(\bm{R},\bm{R}_k)+\frac{\varepsilon}{J}\sum_{j=1}^J d_S^2(\bm{R},\bm{P}_j)
\end{aligned}
\end{equation*}
is given by
\begin{equation*}
\begin{aligned}
\nabla G(\bm{R}) & =\frac{1-\varepsilon}{2K}\sum_{k=1}^K\left(\bm{R}_k^{-1}-\bm{R}^{-1}\bm{R}_k\bm{R}^{-1}\right)\\
&\quad \quad +\frac{\varepsilon}{2J}\sum_{j=1}^J\left(\bm{P}_j^{-1}-\bm{R}^{-1}\bm{P}_j\bm{R}^{-1}\right).
\end{aligned}
\end{equation*}
However, it is more efficient simply by differentiating it about $\varepsilon$ at $\varepsilon=0$; this gives
\begin{equation*}
\begin{aligned}
\bm{H}\left(\sum_{k=1}^K\bm{R}_k^{-1}\right)\overline{\bm{R}} &+\overline{\bm{R}}\left(\sum_{k=1}^K\bm{R}_k^{-1}\right)\bm{H} \\
 & +\frac{m}{n}\sum_{j=1}^J \left( \overline{\bm{R}}\bm{P}_j^{-1}\overline{\bm{R}} - \bm{P}_j\right)=\bm{0},
\end{aligned}
\end{equation*}
which is a continuous Lyapunov equation. It is analytically solvable if the matrix $\left(\sum_{k=1}^K\bm{R}_k^{-1}\right)\overline{\bm{R}}$ or $-\left(\sum_{k=1}^K\bm{R}_k^{-1}\right)\overline{\bm{R}}$ is stable; namely,  its eigenvalues all have either positive or negative real parts. Again as we are only interested in its norm, instead  we take trace and choose
\begin{equation*}
\begin{aligned}
\bm{H}&=\frac{K}{J}\left(\left(\sum_{k=1}^K\bm{R}_k^{-1}\right)\overline{\bm{R}}+\overline{\bm{R}}\left(\sum_{k=1}^K\bm{R}_k^{-1}\right)\right)^{-1}\\
& \quad \quad \quad \times \sum_{j=1}^J\left(\bm{P}_j-\overline{\bm{R}}\bm{P}_j^{-1}\overline{\bm{R}}\right).
\end{aligned}
\end{equation*}
This finishes the proof.

\section{The Riemannian gradient descent algorithm on Stiefel manifolds}
\label{app:SM}
Consider to minimizing a function $\psi(\bm{W})$ defined in the Stiefel manifold $\operatorname{St}(M,\mathbb{C}^N)\subset \mathbb{C}^{N\times M}$. Its Riemannian gradient $\operatorname{grad}\psi(\bm{W})$ is given by  Eq. \eqref{eq:Rgrad}.

Tangent space of the Stiefel manifold is
\begin{equation*}
\begin{aligned}
T_{\bm{W}}\operatorname{St}(M,\mathbb{C}^N)=\left\{\bm{Z}\in\mathbb{C}^{N\times M} \mid \bm{W}^{\operatorname{H}}\bm{Z}+\bm{Z}^{\operatorname{H}}\bm{W}=\bm{0}\right\}.
\end{aligned}
\end{equation*}
Either the Euclidean metric inherited from the embedding in $\mathbb{C}^{N\times M}$ or a canonical metric can be defined on the manifold  $\operatorname{St}(M,\mathbb{C}^N)$. We will be focused on the former one that is defined by
\begin{equation*}
g_{\bm{W}}(\bm{Z}_1,\bm{Z}_2)=\operatorname{tr}\left(\bm{Z}_1^{\operatorname{H}}\bm{Z}_2\right),\quad \bm{Z}_1,\bm{Z}_2\in T_{\bm{W}}\operatorname{St}(M,\mathbb{C}^N).
\end{equation*}
Equipped with the Euclidean metric, a geodesic $\gamma:[0,1]\rightarrow \operatorname{St}(M,\mathbb{C}^N)$ with respect to initial values $\gamma(0)=\bm{W}, \dot{\gamma}(0)=\bm{Z}$
reads
\begin{equation}\label{eq:STgeo}
\gamma(t)=[\bm{W}~  \bm{Z}] \exp\left(t
\left[\begin{array}{cc}
\bm{A} & -\bm{B}\\
\bm{I} & \bm{A}
\end{array}\right]\right)
\left[\begin{array}{c}
\bm{I}\\
\bm{0}
\end{array}\right]\exp(-t\bm{A}),
\end{equation}
where
\begin{equation*}
\bm{A}=\bm{W}^{\operatorname{H}}\bm{Z},\quad \bm{B}=\bm{Z}^{\operatorname{H}}\bm{Z}.
\end{equation*}
The exponential map $\exp:T\operatorname{St}(M,\mathbb{C}^N)\rightarrow \operatorname{St}(M,\mathbb{C}^N)$  is defined by
\begin{equation*}
\exp_{\bm{W}}(\bm{Z})=\gamma(1),
\end{equation*}
where $\gamma$ is given by \eqref{eq:STgeo}.

Now we are ready to propose the Riemannian gradient descent algorithm:
\begin{equation*}
\bm{W}_{l+1}=\operatorname{exp}_{\bm{W}_l}\left(-\eta_l \operatorname{grad}\psi(\bm{W}_l)\right),
\end{equation*}
where $\eta_l$ is the step size.

For more details, see \cite{AMS2008} for instance.

\section{Euclidean gradient of the AIRM distance $d_A^2(f_{\bm{W}}(\bm{X}),f_{\bm{W}}(\bm{Y}))$: Proof of Proposition \ref{prop:AIRMEg}}
\label{app:AIRMEg}

Define
\begin{equation*}
\begin{aligned}
\bm{A}(\varepsilon) &=\left(\left(\bm{W}+\varepsilon\bm{R}\right)^{\operatorname{H}}\bm{X}\left(\bm{W}+\varepsilon\bm{R}\right)\right)^{-1}  \\
&~~~~ \times\left(\left(\bm{W}+\varepsilon\bm{R}\right)^{\operatorname{H}}\bm{Y}\left(\bm{W}+\varepsilon\bm{R}\right)\right),
\end{aligned}
\end{equation*}
where $\bm{X},\bm{Y}\in\mathscr{P}(N,\mathbb{C})$ and $\bm{W},\bm{R}\in \operatorname{St}(M,\mathbb{C}^N)$.
From the definition \eqref{def:gra} and using Eqs. \eqref{eq:logid} and \eqref{eq:intX}, we have
\begin{equation*}
\begin{aligned}
&\left\langle d_A^2(f_{\bm{W}}(\bm{X}),f_{\bm{W}}(\bm{Y})), \bm{R} \right\rangle  =\frac{\operatorname{d}}{\operatorname{d}\!\varepsilon}\Big|_{\varepsilon=0} \operatorname{tr} \left( \operatorname{Log}^2\bm{A}(\varepsilon)\right)\\
&\quad \quad\quad\quad \quad = 2 \operatorname{tr}\left(  \operatorname{Log}\bm{A}(0)\frac{\operatorname{d}}{\operatorname{d}\!\varepsilon}\Big|_{\varepsilon=0} \operatorname{Log} \bm{A}(\varepsilon)\right)\\
&\quad \quad\quad\quad\quad = 2\operatorname{tr}\left( \bm{A}^{-1}(0) \operatorname{Log}\bm{A}(0)  \frac{\operatorname{d}}{\operatorname{d}\!\varepsilon}\Big|_{\varepsilon=0}\bm{A}(\varepsilon)  \right).
\end{aligned}
\end{equation*}
Using the identity that
\begin{equation*}
\frac{\operatorname{d}}{\operatorname{d}\!\varepsilon}\bm{B}^{-1}(\varepsilon)=-\bm{B}^{-1}(\varepsilon)\frac{\operatorname{d}}{\operatorname{d}\!\varepsilon}\bm{B}(\varepsilon)\bm{B}^{-1}(\varepsilon),
\end{equation*}
where $\bm{B}(\varepsilon)$ is invertible, we have
\begin{equation*}
\begin{aligned}
&\left\langle d_A^2(f_{\bm{W}}(\bm{X}),f_{\bm{W}}(\bm{Y})), \bm{R} \right\rangle  \\
&~~~~= 4\operatorname{tr}\bigg( \operatorname{Log}\left(\left(\bm{W}^{\operatorname{H}}\bm{X}\bm{W}\right)^{-1}\left(\bm{W}^{\operatorname{H}}\bm{Y}\bm{W}\right)  \right)\\
&\quad \quad \quad \times\Big( \left(\bm{W}^{\operatorname{H}}\bm{Y}\bm{W}\right)^{-1} \left(\bm{R}^{\operatorname{H}}\bm{Y}\bm{W}\right)\\
&\quad \quad \quad \quad~~~~  -\left(\bm{W}^{\operatorname{H}}\bm{X}\bm{W}\right)^{-1} \left(\bm{R}^{\operatorname{H}}\bm{X}\bm{W} \right)\Big)\bigg).
\end{aligned}
\end{equation*}
Inside the trace, moving the terms $\bm{R}^{\operatorname{H}}\bm{X}\bm{W} $ and $\bm{R}^{\operatorname{H}}\bm{Y}\bm{W} $ to the left of the logarithm finishes the proof.

\bibliographystyle{IEEEtran}
\bibliography{mybibfile}

\end{document}